\documentclass{IEEEtran}
\usepackage{amsmath,amsfonts}
\usepackage{algorithmic}
\usepackage{algorithm}
\usepackage{lineno,hyperref}
\usepackage{multirow}
\usepackage{epstopdf}
\usepackage{graphicx}
\usepackage{float}
\usepackage{caption}
\usepackage{amssymb}
\usepackage{hyperref}
\usepackage{changes}
\usepackage{color,array,amsthm}
\usepackage{graphicx}

\newtheorem{theorem}{Theorem}

\newtheorem{property}{Property}

\newtheorem{remark}{Remark}
\newtheorem{assumption}{Assumption}
\usepackage[inkscapelatex=false]{svg}
\usepackage[numbers,sort&compress]{natbib}
\usepackage{soul} % 导入 soul 包
\usepackage{color, xcolor} % 颜色包，color 必须导入，xcolor 建议导入

\soulregister{\cite}7 % 注册\cite命令
\soulregister{\citep}7 % 注册\citep命令
\soulregister{\citet}7 % 注册\citet命令
\soulregister{\ref}7 % 注册\ref命令
\soulregister{\pageref}7 % 注册\pageref命令
\modulolinenumbers[5]

\begin{document}

\title{Event-Triggered Intermittent Prescribed Performance Control for Spacecraft Attitude Reorientation}

\author{Jiakun Lei,
	Tao Meng,
	Kun Wang,
	Weijia Wang,
	Zhonghe Jin% <-this % stops a space
	\thanks{Jiakun Lei, Ph.D. School of Aeronautics and Astronautics, Zhejiang University, Hangzhou, China, 310027, leijiakun@zju.edu.cn}
	
	\thanks{Tao Meng, Prof., School of Aeronautics and Astronautics, Zhejiang University, Hangzhou, China, 310027; Zhejiang Key Laboratory of micro-nano-satellite, mengtao@zju.edu.cn}
	
	\thanks{Kun Wang, Ph.D., School of Aeronautics and Astronautics, Zhejiang University, Hangzhou, China, 310027, wang\_kun@zju.edu.cn}
	
    \thanks{Weijia Wang, Ph.D., School of Aeronautics and Astronautics, Zhejiang University, Hangzhou, China, 310027, weijiawang@zju.edu.cn}
	
	\thanks{Zhonghe Jin, Prof., School of Aeronautics and Astronautics, Zhejiang University, Hangzhou, China, 310027, Zhejiang Key Laboratory of micro-nano-satellite, jin\_zh@zju.edu.cn}
}

\maketitle

\begin{abstract} 
	This paper focuses on the issue of how to realize spacecraft attitude control with guaranteed performance while conspicuously reducing the actuator acting frequency simultaneously. The prescribed performance control (PPC) scheme is often employed for the control with guaranteed performance. However, conventional PPC controllers are designed from the perspective of continuous system, which contradicts the "discrete" control logic in actual spacecraft control system, and such a problem limited the application value of PPC scheme in actual applications. In order to significantly lower the actuator acting frequency while still maintaining the desired performance, a composite event-trigger mechanism is proposed for this issue, turning off the actuator and eliminating unnecessary control output under appropriate conditions. Further, the proposed composite event-trigger mechanism is combined with a prescribed performance control scheme, constructing a complete controller structure for solving the presented issue. Meanwhile, a special performance function is designed to provide mild transition process. Based on the proposed scheme, a specific backstepping controller is further developed as a validation. Finally, numerical simulation results are presented to validate the effectiveness of the proposed scheme.

\end{abstract}

%\begin{IEEEkeywords}
%\end{IEEEkeywords}

\section{INTRODUCTION}
Recently, modern spacecraft control tasks usually appear with high-performance requirements, such as on-orbit docking and cooperative earth-observation imaging control. Such a requisition strongly motivated the development of the controller built with performance consideration, and the prescribed performance control (PPC) scheme has been of higher research interest recently. Nevertheless, traditional PPC controllers are presented in a continuous viewing, which is contradicts with the operation mode of actual actuators on the spacecraft control system, such as the propulsion system. Such a contradiction limited the application of the PPC scheme to the actual control scenarios, which motivates our idea of developing a prescribed performance control scheme with intermittent control behavior. This paper presents a complete solution for this concerning issue by organically employ the idea derived from intermittent control and combined it with prescribed performance control framework. The proposed composite event-trigger mechanism Meanwhile, the angular velocity norm constraint is also considered, fitting the actual engineering requirements.

Typical controllers are designed without considering performance factors in advance, which indicates that their performance can only be optimized by tedious parameter selection. To solve the problem that guarantees the achievement of the desired performance requirements, the prescribed performance control (PPC) scheme is proposed by Bechlioulis and Rovathakis in \cite{bechlioulis_adaptive_2009}. Since the PPC control scheme can be directly enforced to many existing frameworks, it has been employed in the attitude control problem in many works and showing its capability in handling the controller design with performance requirements, as stated in \cite{wei2018learning,hu2018adaptive,WANG2022,yong2020flexible}. In our previous work \cite{lei2022robust}, we present a PPC-based structure that can realize the desired-performance control while ensuring a low angular velocity during the attitude maneuver process simultaneously, which is realized by constructing a new kind of Barrier Lyapunov function (BLF).

Although traditional PPC controllers effectively handle the performance requirements, they require high-frequency control during the whole control procedure. Such an acting behavior is hard to bear for actual spacecraft actuators. Therefore, it strongly motivated our investigation for realizing the PPC with practically low-frequency control.

In order to lower the acting frequency of the actuator, controllers designed with intermittent acting behavior is considered in this paper. According to existing literature, 
controllers designed intermittently can be generally classified into three types, which treat intermittent control from the perspective of the switching system, the perspective of Lyapunov analysis, and the perspective of the event-trigger mechanism. The first one considers the system as a switched one, and the stability analysis is performed based on it, as stated in \cite{li2007stabilization,wang2016stabilization}. For the intermittent control scheme based on Lyapunov analysis, it considers the influence caused by the intermittent control behavior, quantitatively implements such an influence as an additional Lyapunov residual term, and the system stability is analyzed based on ensuring its practical boundedness. Such an idea can be found in \cite{gawthrop2009event,gawthrop2011intermittent,gawthrop2015intermittent,loram2011human}. These two control schemes are effective for linear system analysis but unsuitable for nonlinear system control. Since the spacecraft attitude system is strongly nonlinear with strong coupling characteristics, these two typical methods are hard to employ in the PPC controller design.

Unlike those previous works, the third kind of framework handles intermittent control in an event-triggered way, which is an inspirational work recently presented by Pio Ong \cite{ong2022stability}. Two different trigger mechanisms are introduced in the so-called intermittent event-trigger control scheme, responsible for the turn-on and the turn-off of the actuator. The author utilizes this framework to realize safety-critical control, introducing it to the control-barrier function-based structures.

Motivated by the idea lies in the intermittent-event trigger mechanism, this paper investigated the possibility of enforces the intermittent event-trigger control methodology and the prescribed performance control scheme together, aiming at realizing the desired performance with much lesser actuator activation.
In \cite{ong2022stability}, the intermittent event-trigger control scheme is introduced to a stabilization safety-critical problem. However, since spacecraft need accelerating and braking such that the spacecraft is of high-dynamic status during the control procedure, some specific designs should be reconsidered in advance.

In order to make the intermittent event trigger suitable for attitude reorientation control, a composite event-trigger mechanism is redesigned with an exponential function-based trigger mechanism and a performance evaluation function-based trigger mechanism.
Further, we enforce this proposed trigger mechanism into a prescribed performance control (PPC) scheme stated in our previous work \cite{lei2022robust}. Besides, a special performance function is proposed based on the hyperbolic-tangent function, providing mild actuator behavior during the attitude maneuver stage. In this paper, detailed suggestions for building the intermittent event-trigger attitude control are analyzed and presented, facilitating its application to various control scenarios. Based on these modifications and analysis, a complete controller structure is presented that is able to achieve the guaranteed performance with a much lesser acting frequency of the actuator, which is of potential application value in the real spacecraft control system.

The following paper is organized as follows: section [\ref{pre}] presents some fundamental introduction about notations and system modeling under event-trigger mechanisms. The main results of this paper are presented in section [\ref{results}], while the specific controller derivation is elaborated in section [\ref{controller}] thoroughly. Theoretical analysis of the system stability is presented in section [\ref{response}], and the numerical simulation results are illustrated and analyzed in section [\ref{sim}] accordingly.

\section{PRELIMINARIES}\label{pre}
\subsection{Notations}
In this subsection, we define the following symbols for convenient. $\|\cdot\|$ stands for the Euclidean norm of arbitrary vector, or the induced norm of arbitrary given matrix. For any given vector $\boldsymbol{a}$, $a_{i}$ represents the $i$ th component of $\boldsymbol{a}$. Correspondingly, $\boldsymbol{a}^{\times}$ denotes the $\mathbb{R}^{3\times 3}$ cross manipulation matrix for $\boldsymbol{a}$ such that $\boldsymbol{a} \times \boldsymbol{b} = \boldsymbol{a}^{\times}\boldsymbol{b}$ is always satisfied. For the spanned diagonal matrix $\left(a_{i}\right)_{d}$, it denotes the $\mathbb{R}^{3\times 3}$ diagonal matrix that is consist of $a_{1}, a_{2},...a_{i}$. Similarly, $\left(a_{i}\right)_{\text{v}}$ stands for the spanned column-vector of $a_{i}$ as $\left(a_{i}\right) = \left[a_{1},a_{2},...a_{i}\right]^{\text{T}}$. The maxima and the minima is represented as $\left(\cdot\right)_{\text{max}}$ and $\left(\cdot\right)_{\text{min}}$ in the following analysis.

\subsection{System Modeling}\label{secsys}
Considering the general rigid-body attitude error system modeled with unit quaternion, for the attitude reorientation problem such that the desired angular velocity $\boldsymbol{\omega}_{d} = \boldsymbol{0}$ holds, the angular velocity can be expressed as $\boldsymbol{\omega}_{e} = \boldsymbol{\omega}_{s} - \boldsymbol{0}$. Hence, the system can be modeled as follows \cite{walls_globally_2005}:
\begin{equation} 
	\centering
	\label{errorsystem}
	\begin{aligned}
		\dot{\boldsymbol{q}}_{ev} &= \boldsymbol{F}_{s}\boldsymbol{\omega}_{s}\quad
		\dot{q}_{e0} = -\frac{1}{2}\boldsymbol{q}^{\text{T}}_{ev}\boldsymbol{\omega_{s}} \\
		\boldsymbol{J}\dot{\boldsymbol{\omega}}_{s} &=  -\boldsymbol{\omega}^{\times}_{s}\boldsymbol{J}\boldsymbol{\omega}_{s} + \boldsymbol{u}_{\text{act}} + \boldsymbol{d}
	\end{aligned}
\end{equation}
where the attitude error quaternion is represented as $\boldsymbol{q}_{e} = \left[\boldsymbol{q}_{ev}\quad q_{e0}\right]^{\text{T}}\in\mathbb{R}^{4}$, and its $i$ th component is denoted as $q_{ei}$ correspondingly. $\boldsymbol{q}_{ev} = \left(q_{ei}\right)_{\text{v}}\left(i = 1,2,3\right)$ is a $\mathbb{R}^{3}$ column vector represents the vector part of the attitude error quaternion $\boldsymbol{q}_{e}$, and $q_{e0}$ stands for the scalar part of $\boldsymbol{q}_{e}$. $\boldsymbol{\omega}_{s}$ denotes the attitude angular velocity of the spacecraft with respect to the inertial frame expressed in the body-fixed frame. $\boldsymbol{J}$ represents the inertial tensor matrix expressed in the body-fixed frame, which is a $\mathbb{R}^{3\times 3}$ positive-definite matrix all along. $\boldsymbol{u}_{\text{act}}$ denotes the actual exerted control input, equivalent to $\boldsymbol{u}\left(t_{k}\right)$, where $t_{k}$ represents the triggering time instant. $\boldsymbol{d}$ stands for the lumped external disturbance, while $\boldsymbol{F}_{s}$ is the Jacobian matrix defined as $\boldsymbol{F}_{s} = \frac{1}{2}\left[q_{e0}\cdot\boldsymbol{I}_{3} + \boldsymbol{q}^{\times}_{ev}\right]$.

Further, we consider the system model under event-trigger mechanism.
The actuator output will only be updated at the trigger time instant for the system with an event-trigger mechanism. Accordingly, a deviation exists between the calculated control command and the actual exerted control effort. For arbitrary trigger time instant $t_{k}$, define an inline error state variable $\boldsymbol{e}_{u}$ as $\boldsymbol{e}_{u} = \boldsymbol{u}\left(t\right) - \boldsymbol{u}\left(t_{k}\right)$, the system dynamics equation in (\ref{errorsystem}) can be reformed as \cite{wu2018event}:
\begin{equation}\label{errdyna}
	\boldsymbol{J}\dot{\boldsymbol{\omega}}_{s} = -\boldsymbol{\omega}^{\times}_{s}\boldsymbol{J}\boldsymbol{\omega}_{s} + \boldsymbol{u}\left(t\right) - \boldsymbol{e}_{u}\left(t\right) + \boldsymbol{d}
\end{equation}
notably, $\boldsymbol{u}\left(t\right)$ denotes the calculated control command signal. The following analysis will be performed based on this model.

\subsection{Assumptions}
In this paper we make the following assumptions for the analysis of the controller.
\begin{assumption}
	The inertial matrix is assumed to be a known one in this paper, which is a positive-definite matrix all along. Since the inertial uncertainties is not the main focus of this paper, such an assumption is rational.
\end{assumption}
\begin{assumption}
The external disturbance is unknown, but is bounded by a known constant, which is rational for a common attitude control task. Accordingly, we assume that $\|\boldsymbol{d}\| \le D_{m}$.
\end{assumption}

\section{CONTROL SCHEME DESIGN}\label{results}
Inspired by the work stated in \cite{ong2022stability}, we design a composite event-trigger mechanism governed by two independent trigger mechanisms. These two trigger mechanisms take responsible for the turn-on and the turn-off of the actuator, cutting off the unnecessary controller output and conspicuously reducing the acting frequency of the actuator. Further, the designed composite trigger mechanism is introduced into a prescribed performance control (PPC) framework that we presented previously in \cite{lei2022robust}, providing guaranteed performance control under limited actuator acting. A newly-designed special performance function is introduced in order to provide desired characteristic.

This section is organized as follows: the designing of the PPC scheme is stated in subsection [\ref{PF}],[\ref{err}] and [\ref{BLF}], corresponding to the introduction of the designed performance function, error transformation, and the designed BLF function respectively. Further, the designed event-trigger mechanism is presented in section [\ref{trigoff}] and [\ref{trigon}], corresponding to the turn-off and the turn-on trigger mechanism.

\subsection{Performance Function Design}\label{PF}
As stated in \cite{lei2022singularity}, in the typical prescribed performance control (PPC) scheme, the state variable is expected to be restricted in a sector-like region enclosed by the so-called performance function, which is designed according to the performance requirement.
In this paper, we design the performance function as follows:
\begin{equation}
	\rho\left(t\right) = \frac{\rho_{0}+\rho_{\infty}}{2} - \left(\frac{\rho_{0}-\rho_{\infty}}{2}\right)\tanh\left(\frac{T-T_{s}}{f_{s}}\right)
\end{equation}
where $\rho_{0}$ and $\rho_{\infty}$ denotes the initial value and the terminal value of $\rho\left(t\right)$, respectively. The parameter constant $T_{s}$ and $f_{s}$ is related to how fast and when the performance function started decreasing.

The main purpose of designing such a performance function is to alleviate the burden of the controller. According to the existing literature, the exponential function is the mostly applied one for the performance function. However, it decreases rapidly at the beginning of the control, and its decreasing rate will become slower as it approaches the steady state. Such a characteristic is not suitable for the proposed scheme. Since $\boldsymbol{\omega}_{s}$ needs some time to track the virtual control law $\boldsymbol{\alpha}$, hence the output state trajectory may not response as desired at first. Therefore 
there should be some additional space in the constraint region for the state trajectory at the beginning.
On the other hand, since the angular velocity will be restricted to a relatively small value as we expected, hence the convergence of the state trajectory will be much slower than a typical exponentially-converged one. Accordingly, the constraint region should not converging too fast at the initial stage. 

Following these stated reasons, we design the stated performance function to provide enough space for the state trajectory, providing relatively mild attitude maneuver progress and controller output.

\subsection{Error Transformation Procedure}\label{err}
For arbitrary error state variable denoted as $\boldsymbol{e}\left(t\right) = \left[e_{1}\left(t\right)...e_{i}\left(t\right)\right]^{\text{T}}$, defining the corresponding performance function for $e_{i}\left(t\right)$ as $\rho_{i}\left(t\right)$, the desired state constraint can be expressed as follows:
\begin{equation}\label{cons}
	-\rho_{i}\left(t\right) < e_{i}\left(t\right) < \rho_{i}\left(t\right)
\end{equation}
That is, $e_{i}\left(t\right)$ will be constrained into a symmetric funnel-like region.
Further, defining a transformed variable as $\varepsilon_{i} = e_{i}/\rho_{i}$, the state constraint in equation (\ref{cons}) can be rewritten as $|\varepsilon_{i}\left(t\right)| < 1$. Notably, if  $\boldsymbol{\varepsilon}^{\text{T}}\boldsymbol{\varepsilon} < 1$ is held, the constraint (\ref{cons}) will be satisfied. Hence we consider the tight one instead of the original one; i.e., we consider the state constraint as $\|\boldsymbol{\varepsilon}\|^{2} < 1$.

Before going further, note we have the following property.
\begin{property}
	\label{P2}
	Considering arbitrary $k>0$, then must exists a big enough constant $c$ and a constant $k_{0}$ satisfies $k_{0}<k$ such that $k\tanh\left(cx\right) \ge x$ will be held on $x\in\left[0,+k_{0}\right)$.
\end{property}
According to this property, a inference can be presented: if a practical boundary for $x$ exists such that $|x|\le k_{0}$ can be satisfied. Then considering a corresponding constant $k$ near $k_{0}$ while holds $k > k_{0}$, $|k\tanh\left(cx\right)| \ge |x|$ will be hold on $x\in\left[-k_{0},k_{0}\right]$ if $c$ is big enough.

\subsection{Barrier Lyapunov Function (BLF) Design}\label{BLF} 
Based on our previous work \cite{lei2022robust}, the Barrier Lyapunov Function (BLF) is designed as follows:
\begin{equation}
	V_{\text{BLF}} = \frac{k_{1}}{2}F_{1}\ln\left[\cosh\left(\frac{\boldsymbol{\varepsilon}^{\text{T}}\boldsymbol{\varepsilon}}{F_{1}}\right)\right]
\end{equation}
where $k_{1}$, $F_{1}$ are positive design parameters that need indicating. Taking the partial-derivative of the presented BLF with respect to $\boldsymbol{\varepsilon}$, it can be obtained that $\frac{\partial V_{\text{BLF}}}{\partial \boldsymbol{\varepsilon}} = k_{1}\tanh\left[\|\boldsymbol{\varepsilon}\|^{2}/F_{1}\right]\boldsymbol{\varepsilon}^{\text{T}}$. This indicates that the gradient with respect to $\boldsymbol{\varepsilon}$ can be modified by explicitly changing $k_{1}$. Notably, considering a function expressed as $H\left(x\right) = \ln\left[\cosh\left(x\right)\right]$, one has the following property:
\begin{property}\label{P1}
	For the defined function $H\left(x\right)$ on $x\in\left[0,+\infty\right)$, we have $\frac{1}{2}x\tanh x \le H\left(x\right) \le x\tanh x$.
\end{property}

\subsection{Event-Trigger Mechanism Design}

\label{sectrigoff}
Firstly, we consider the following problem: if the controller remains on a sample-and-hold status, when should we \textbf{turn off} the controller?
Considering that the new information had just been updated to the actuator, the measurement error will start accumulating continuously from zero until the next triggering condition is satisfied. Therefore, it is a natural idea that the actuator should be turned off when the held output value can no longer ensure the strict convergence of the system. 

Based on such an idea and considering the real limitations in real engineering work, two factors will determine the next trigger time: the following triggering condition and the maximum output time limitation. Accordingly, we design the turn-off trigger mechanism as follows:
\begin{equation}
	\label{trigoff}
	\begin{aligned}
		      t^{\text{act}}_{k} &= \inf_{t>t^{\text{on}}_{k}}\left\{\|\boldsymbol{e}_{u}\|^{2} =s e^{-\beta t} + m\right\}\\
		      t^{\text{pas}}_{k} &= t^{\text{on}}_{k} + T_{\text{max}}\\
		      t^{\text{off}}_{k} &= \min\left(t^{\text{act}}_{k},t^{\text{pas}}_{k}\right)
	\end{aligned}
  \end{equation}
where $t^{\text{on}}_{k} > 0$ stands for the $k$ th trigger time instant of the turn-on mechanism, $t^{\text{off}}_{k}$ represents the $k$ th trigger time that turns off the actuator. $s>0$, $\beta>0$ and $m>0$ are design parameters need indicating. $T_{\text{max}}$ represents the allowed maximum time duration for continuously actuator output. $\boldsymbol{e}_{u}\in\mathbb{R}^{3}$ denotes the accumulated error defined as $\boldsymbol{e}_{u} = \boldsymbol{u}\left(t\right) - \boldsymbol{u}\left(t^{\text{on}}_{k}\right)$, as stated in subsection [\ref{secsys}]. Detailed proof of the turn-off trigger mechanism will be discussed in section [\ref{response}] later.

\label{sectrigon}
Subsequently, we discuss when to \textbf{turn on} the actuator. When the controller had been turned off at the time instant $t = t^{\text{off}}_{k}$, due to the lack of control, the system may diverge from the desired state. However, since the previous trigger condition (\ref{trigoff}) guarantees that the system will be strictly converged (proved in theorem [\ref{T1}] later), there may exist an additional margin that allows the system to diverge temporarily, hence it may be unnecessary to re-control the system immediately.

 Motivated by the prescribed performance control (PPC) scheme and the work stated in \cite{velasco2009lyapunov}, we use a performance-evaluation function to qualitatively describe this process: although the strict convergence of the Lyapunov trajectory may not be ensured at every moment without control. However, by restricting it remains beneath a function that is monotonically converged to the origin (or a small region near the origin), the Lyapunov trajectory will be able to guide to the steady state, and the stability of the system can be practically obtained.

Based on such an idea, defining a performance-evaluation function as $S\left(t\right)$. For arbitrary candidate Lyapunov function as $V\left(t\right)$, if the inequality $S\left(t\right) > V\left(t\right)$ can be satisfied, then the system will be asymptotically-stabled if $S\left(t\right)$ is an asymptotically converged one. Note that exponential convergence is a common characteristic of the system stability, hence it is natural to design the performance-evaluation function as follows:
\begin{equation}
	S\left(t\right) = \left(S_{0} - S_{\infty}\right)e^{-kt} + S_{\infty}
\end{equation}
where $S_{0}$ and $S_{\infty}$ stands for the initial value and the terminal value of $S\left(t\right)$, that is, $S\left(t\right)$ will converge from $S_{0}$, and finally achieve a small region near $S_{\infty}$. $k$ stands for the convergence rate of $S\left(t\right)$. Accordingly, we design the turn-on trigger mechanism as follows:

\begin{equation}
	\label{trigon}
	t^{\text{on}}_{k+1} = \inf_{t>t^{\text{off}}_{k}}\left\{S\left(t\right)-V\left(t\right) =\delta_{m}\right\}
\end{equation}
where $\delta_{m} > 0$ is an additional robust buffer boundary need indicating. It can be observed that the trigger mechanism (\ref{trigon}) guarantees that $S\left(t\right) > V\left(t\right)$ will be satisfied.

\section{CONTROLLER DERIVATION}\label{controller}
 This subsection presents a complete controller structure by utilizing the backstepping designing methodology, the presented intermittent event-trigger mechansism and the PPC control scheme. The outer loop of the controller is designed to be triggered by the aforementioned event-trigger mechanism stated in section [\ref{sectrigoff}] and [\ref{sectrigon}]. A sketch map of the system controller diagram is illustrated in the figure [\ref{fig_sys}] as below.
 
\begin{figure*}
	\centering 
	\includegraphics[scale = 0.6]{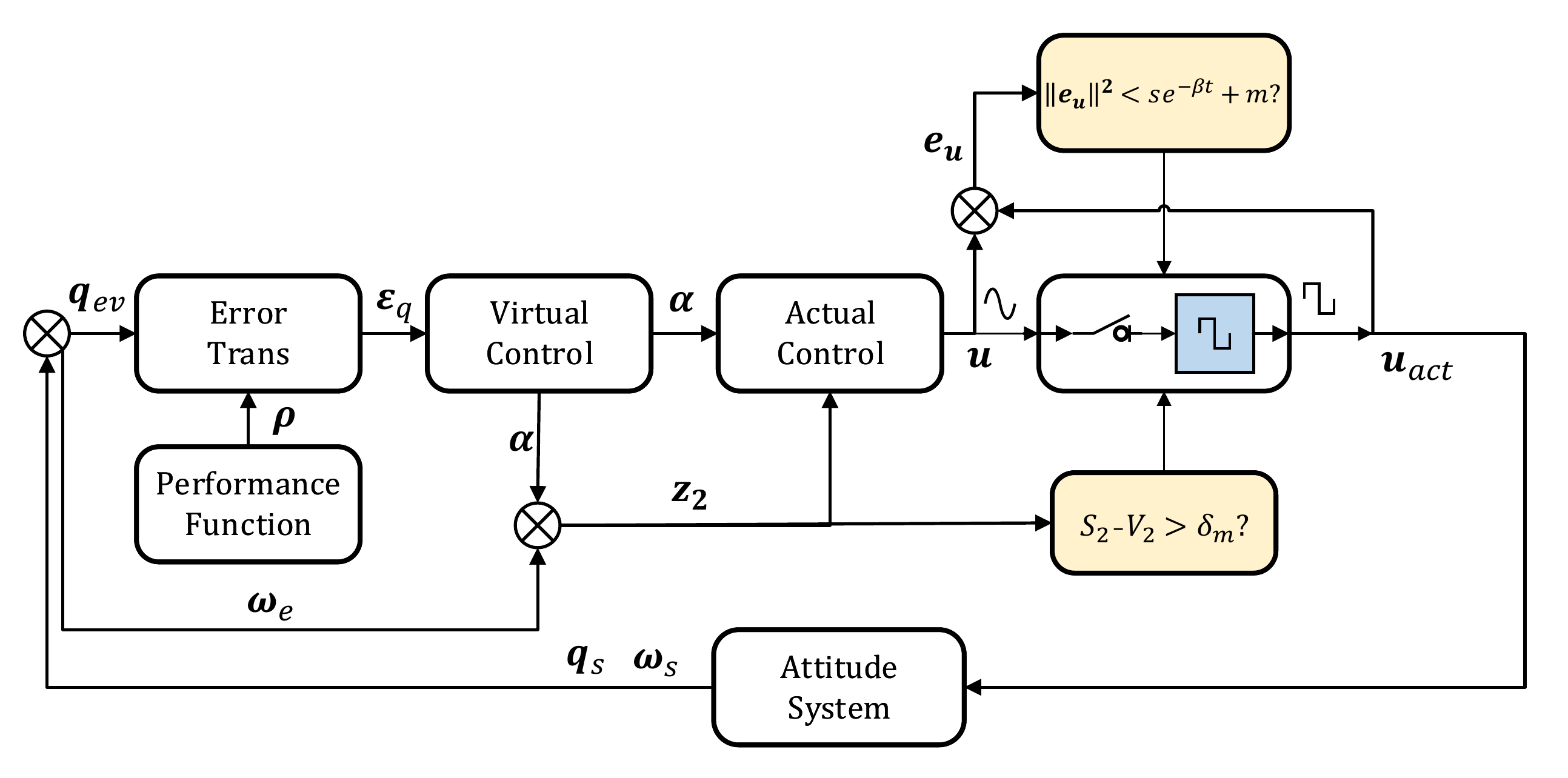}
	\caption{Diagram of the Controller}    
	\label{fig_sys}  
\end{figure*}

Define the translated variable of $\boldsymbol{q}_{ev}$ as $\boldsymbol{\varepsilon}_{q}$, the $i$ th component of $\boldsymbol{q}_{ev}$ is denoted as $q_{evi}$. The corresponding performance function of $q_{evi}$ is represented as $\rho_{i}$. According to the aforementioned definition in subsection \ref{err}, the $i$ th component of $\boldsymbol{\varepsilon}_{q}$ can be expressed as $\varepsilon_{qi} = q_{evi}/\rho_{i}$. Taking the time-derivative of $\boldsymbol{\varepsilon}_{q}$, one can be obtained that:
\begin{equation}\label{depsilon}
	\dot{\boldsymbol{\varepsilon}}_{q} = \boldsymbol{\psi}_{q}\left(\dot{\boldsymbol{q}}_{ev} - \boldsymbol{\eta}_{q}\boldsymbol{q}_{ev}\right)
\end{equation}
where $\boldsymbol{\psi}_{q} = \left(1/\rho_{i}\right)_{\text{d}}$ and $\boldsymbol{\eta}_{q} = \left(\dot{\rho}_{i}/\rho_{i}\right)_{\text{d}}$ are $\mathbb{R}^{3\times 3}$ diagonal matrices that is defined accordingly.
Considering the relationship in attitude kinematics equation (\ref{errorsystem}), we have $\dot{\boldsymbol{q}}_{ev} = \boldsymbol{F}_{s}\boldsymbol{\omega}_{s}$. Substituting this relationship into equation (\ref{depsilon}) can further obtain that: $\dot{\boldsymbol{\varepsilon}}_{q} = \boldsymbol{\psi}_{q}\left(\boldsymbol{F}_{s}\boldsymbol{\omega}_{s} - \boldsymbol{\eta}_{q}\boldsymbol{q}_{ev}\right)$.

Firstly, defining the first-layer error subsystem as $\boldsymbol{z}_{1} = \boldsymbol{\varepsilon}_{q}$. Viewing the angular velocity as the virtual control law, denoted as $\boldsymbol{\alpha}$, we design the virtual control law $\boldsymbol{\alpha}$ as follows:
\begin{equation}
	\label{virtual}
	\boldsymbol{\alpha} = -\frac{|q_{e0}|}{2}k_{m}M_{\omega}\boldsymbol{F}^{-1}_{s}\boldsymbol{\psi}_{q}^{-1}\left(\tanh\gamma \varepsilon_{qi}\right)_{\text{v}}
\end{equation}
where $k_{m} > 0$ and $\gamma > 0$ are parameter constants to be designed, $M_{\omega}$ is a parameter related to the upper boundary of $\|\boldsymbol{\alpha}\|$.
\begin{remark}\label{boundedv}
	Taking the norm of  $\boldsymbol{\alpha}$, it can be observed that $ \boldsymbol{\alpha} \le \frac{|q_{e0}|}{2}k_{m}M_{\omega}\|\boldsymbol{F}^{-1}_{s}\boldsymbol{\psi}^{-1}_{q}\|\le\frac{|q_{e0}|}{2}k_{m}M_{\omega}\|\boldsymbol{F}^{-1}_{s}\|\|\boldsymbol{\psi}^{-1}_{q}\|$. Notably, since $\|\boldsymbol{F}^{-1}_{s}\| = \frac{2}{|q_{e0}|}$, hence we can yield $\| \boldsymbol{\alpha}\| \le k_{m}M_{\omega}\left(|\rho_{i}|\right)_{\text{max}}$. This indicates that the desired angular velocity for the spacecraft is strongly bounded by $k_{m}M_{\omega}$, which can be designed straightforwardly. 
	
	On the other hand, if the actual angular velocity $\boldsymbol{\omega}_{s}$ could follow the virtual control law $\boldsymbol{\alpha}$ tightly, the angular rate of the spacecraft will be globally limited to the desired value, and the given constraint will be satisfied. For the following analysis, we define the practical boundary for $\|\boldsymbol{\omega}_{s}\|$ as $B_{\omega}$.
\end{remark}
\begin{remark}\label{boundedAlpha}
	We discuss the boundedness of $\dot{\boldsymbol{\alpha}}$ and $\ddot{\boldsymbol{\alpha}}$. In view of the physical meaning of $\boldsymbol{\alpha}$, it represents an angular velocity that is able to lead the system converge rapidly. Further, as stated in remark \ref{boundedv}, $\|\boldsymbol{\alpha}\|$ is strict bounded by constant. Therefore, since $\dot{\boldsymbol{\alpha}}$ stands for such an virtual angular acceleration, it is a natural inference that $\|\dot{\boldsymbol{\alpha}}\|$ will be bounded by a constant, which will be denoted as $B_{\alpha}$ in the following analysis. Thing will be similar for $\|\ddot{\boldsymbol{\alpha}}\|$ such that exist a constant satisfies $\|\ddot{\boldsymbol{\alpha}}\|\le B_{2\alpha}$.
\end{remark}

In view of the typical backstepping methodology, we define another error subsystem as $\boldsymbol{z}_{2} = \boldsymbol{\omega}_{s} - \boldsymbol{\alpha}$. Regarding to the dynamics equation (\ref{errdyna}) and combining with the inline error variable $\boldsymbol{e}_{u}$, one can be yielded that:
\begin{equation}\label{systemtrans}
	\boldsymbol{J}\dot{\boldsymbol{z}}_{2} = -\boldsymbol{\omega}^{\times}_{s}\boldsymbol{J}\boldsymbol{\omega}_{s} + \boldsymbol{u}\left(t\right) - \boldsymbol{e}_{u} + \boldsymbol{d} - \boldsymbol{J}\dot{\boldsymbol{\alpha}}
\end{equation}
$t$ will be omitted for brevity. 
To stabilize the $\boldsymbol{z}_{2}$ error subsystem, we design the following actual control law $\boldsymbol{u}$, stated as below:
\begin{equation}
	\label{controllaw}
	\boldsymbol{u} = \boldsymbol{\omega}^{\times}_{s}\boldsymbol{J}\boldsymbol{\omega}_{s} - K_{2}\boldsymbol{z}_{2} + \boldsymbol{J}\dot{\boldsymbol{\alpha}} - \hat{\boldsymbol{d}} - \boldsymbol{P}_{q}
\end{equation}
 $\hat{\boldsymbol{d}}\in\mathbb{R}^{3}$ is a rough compensation for external disturbance expressed as $\hat{\boldsymbol{d}} = D_{m}\left(\tanh\left(z_{2i}/p_{i}\right)\right)_{\text{v}}$, where $p_{i} > 0 \left(i = 1,2,3\right)$ are design parameters that need indicating. $\boldsymbol{P}_{q} = k_{1}\tanh\left(\boldsymbol{\varepsilon}^{\text{T}}_{q}\boldsymbol{\varepsilon}_{q}/F_{1}\right)\boldsymbol{\psi}_{q}\boldsymbol{F}_{s}\boldsymbol{\varepsilon}_{q}\in\mathbb{R}^{3}$ is a dynamical compensation term for the first $\boldsymbol{z}_{1}$ layer.
 \begin{remark}\label{boundedeps}
 	In the following analysis, we regard the translated variable to be a practical-bounded one such that $|\varepsilon_{qi}| \le k_{u}$ can be satisfied. On the one hand, since $q_{evi}$ trajectory will be enclosed by the performance constraint region at the steady state, $|\varepsilon_{qi}| < 1$ will be satisfied. On the other hand, although there exists possibility that $q_{evi}$ trajectory going out of the constraint region temporarily during the convergence stage. However, due to the designed special performance function and the fact that $\varepsilon_{i}$ is defined in a proportional form, the deviation caused by $q_{evi}$ will not have much influence on the translated variable during the convergence stage. In this way, it is rational to assume that there exists a practical boundary for each $\varepsilon_{qi} \left(i = 1,2,3\right)$.
 \end{remark}

 \section{STABILITY ANALYSIS}\label{response}
 Based on the aforementioned section \ref{BLF}, we choose a candidate Lyapunov function $V_{1}$ as follows:
 \begin{equation}
 	    V_{1} = \frac{k_{1}}{2}F_{1}\ln\left[\cosh\left(\boldsymbol{\varepsilon}^{\text{T}}_{q}\boldsymbol{\varepsilon}_{q}/F_{1}\right)\right]
 \end{equation}
 Taking the time-derivative of $V_{1}$, one can be obtained that:
 \begin{equation}
 	   \begin{aligned}
 	   	    \dot{V}_{1} = k_{1}\tanh\left(\boldsymbol{\varepsilon}^{\text{T}}_{q}\boldsymbol{\varepsilon}_{q}/F_{1}\right)\boldsymbol{\varepsilon}^{\text{T}}_{q}\dot{\boldsymbol{\varepsilon}}_{q}
 	   \end{aligned}
 \end{equation}
considering the expression of the time-derivative of $\boldsymbol{\varepsilon}_{q}$, we can yield the following result:
\begin{equation}
	\begin{aligned}
		\dot{V}_{1} &=  k_{1}\tanh\left(\boldsymbol{\varepsilon}^{\text{T}}_{q}\boldsymbol{\varepsilon}_{q}/F_{1}\right)\boldsymbol{\varepsilon}^{\text{T}}_{q}\left[\boldsymbol{\psi}_{q}\boldsymbol{F}_{s}\boldsymbol{\omega}_{s} -\boldsymbol{\psi}_{q}\boldsymbol{\eta}_{q}\boldsymbol{q}_{ev}\right]	
	\end{aligned}
\end{equation}
Further, substituting the virtual control law (\ref{virtual}) into $\dot{V}_{1}$, it can be obtained that:
\begin{equation}
	\label{dotV1}
	\begin{aligned}
		\dot{V}_{1} &= 
		-\frac{|q_{e0}|k_{m}M_{\omega}}{2}k_{1}\tanh\left(\boldsymbol{\varepsilon}^{\text{T}}_{q}\boldsymbol{\varepsilon}_{q}/F_{1}\right)\boldsymbol{\varepsilon}^{\text{T}}_{q}\left(\tanh\gamma\varepsilon_{qi}\right)_{\text{v}}\\
		&\quad+k_{1}\tanh\left(\boldsymbol{\varepsilon}^{\text{T}}_{q}\boldsymbol{\varepsilon}_{q}/F_{1}\right)\boldsymbol{\varepsilon}^{\text{T}}_{q}\boldsymbol{\psi}_{q}\boldsymbol{F}_{s}\boldsymbol{z}_{2}\\
		&\quad -k_{1}\tanh\left(\boldsymbol{\varepsilon}^{\text{T}}_{q}\boldsymbol{\varepsilon}_{q}/F_{1}\right)\boldsymbol{\varepsilon}^{\text{T}}_{q}\boldsymbol{\psi}_{q}\boldsymbol{\eta}_{q}\boldsymbol{q}_{ev}	
	\end{aligned}
\end{equation}
Notably, considering the part expressed as $\boldsymbol{\varepsilon}^{\text{T}}_{q}\left(\tanh\gamma\varepsilon_{qi}\right)_{\text{v}}$, owing to the fact depicted in property \ref{P2}, one can be obtained that $k_{m}\left(\tanh\gamma\varepsilon_{qi}\right) \ge \varepsilon_{qi}$ will be hold for $\varepsilon_{qi}\in\left(0,k_{u}\right)$, where $k_{u}$ is a practical bound for $\varepsilon_{qi}$ mentioned in remark [\ref{boundedeps}]. Accordingly, we have $k_{m}\boldsymbol{\varepsilon}_{q}^{\text{T}}\left(\tanh\gamma\varepsilon_{qi}\right)_{\text{v}} \ge \boldsymbol{\varepsilon}^{\text{T}}_{q}\boldsymbol{\varepsilon}_{q}$.

Combining this result with equation (\ref{dotV1}), it can be further obtained that:
\begin{equation}
	\begin{aligned}
		&-\frac{|q_{e0}|k_{m}M_{\omega}}{2}k_{1}\tanh\left(\boldsymbol{\varepsilon}_{q}^{\text{T}}\boldsymbol{\varepsilon}_{q}/F_{1}\right)\boldsymbol{\varepsilon}^{\text{T}}_{q}\left(\tanh\gamma\varepsilon_{qi}\right)_{\text{v}}\\
		\le& 			-\frac{|q_{e0}|M_{\omega}}{2}k_{1}\tanh\left(\boldsymbol{\varepsilon}_{q}^{\text{T}}\boldsymbol{\varepsilon}_{q}/F_{1}\right)\boldsymbol{\varepsilon}^{\text{T}}_{q}\boldsymbol{\varepsilon}_{q}
	\end{aligned}
\end{equation}
Further, according to the conclusion in property (\ref{P1}), we have $\tanh\left(\boldsymbol{\varepsilon}^{\text{T}}_{q}\boldsymbol{\varepsilon}_{q}/F_{1}\right)\boldsymbol{\varepsilon}_{q}^{\text{T}}\boldsymbol{\varepsilon}_{q} \le F_{1}\ln\left[\cosh\left(\boldsymbol{\varepsilon}^{\text{T}}_{q}\boldsymbol{\varepsilon}_{q}/F_{1}\right)\right]$.
 Accordingly, this inequality can be further relaxed to the following form:
\begin{equation}\begin{aligned}
	&\quad-\frac{|q_{e0}|M_{\omega}}{2}k_{1}\tanh\left(\boldsymbol{\varepsilon}^{\text{T}}_{q}\boldsymbol{\varepsilon}_{q}/F_{1}\right)\boldsymbol{\varepsilon}^{\text{T}}_{q}\boldsymbol{\varepsilon}_{q}\\
	&\le   -|q_{e0}|M_{\omega}\frac{k_{1}}{2}F_{1}\ln\left[\cosh\left(\frac{\boldsymbol{\varepsilon}_{q}^{\text{T}}\boldsymbol{\varepsilon}_{q}}{F_{1}}\right)\right]\\
	&=  -|q_{e0}|M_{\omega}V_{1}
\end{aligned}
\end{equation}
Subsequently, considering the third term in expression (\ref{dotV1}), note the following relationship exists: $\boldsymbol{\varepsilon}_{q} = \boldsymbol{\psi}_{q}\boldsymbol{q}_{ev}$.
 Hence, it can be further rewritten into the following form as  $-k_{1}\tanh\left(\boldsymbol{\varepsilon}^{\text{T}}_{q}\boldsymbol{\varepsilon}_{q}/F_{1}\right)\boldsymbol{\varepsilon}_{q}^{\text{T}}\boldsymbol{\eta}_{q}\boldsymbol{\varepsilon}_{q}$. According to the definition, $\boldsymbol{\eta}_{q} = \left(\dot{\rho}_{i}/\rho_{i}\right)_{\text{d}}$ will be hold. Therefore, it can be inferred that:
 \begin{equation}
 	\begin{aligned}
 		 	     \quad&-k_{1}\tanh\left(\boldsymbol{\varepsilon}^{\text{T}}_{q}\boldsymbol{\varepsilon}_{q}/F_{1}\right)\boldsymbol{\varepsilon}^{\text{T}}_{q}\boldsymbol{\eta}_{q}\boldsymbol{\varepsilon}_{q} \\
 		 	     \le& \left(|\dot{\rho}_{i}|/\rho_{i}\right)_{\text{max}}k_{1}\tanh\left(\boldsymbol{\varepsilon}^{\text{T}}_{q}\boldsymbol{\varepsilon}_{q}/F_{1}\right)\boldsymbol{\varepsilon}^{\text{T}}_{q}\boldsymbol{\varepsilon}_{q}
 	\end{aligned}
 \end{equation}
Notably, owing to the conclusion in property (\ref{P1}), we have:
\begin{equation}
	\frac{k_{1}}{4}\tanh\left(\boldsymbol{\varepsilon}^{\text{T}}_{q}\boldsymbol{\varepsilon}_{q}/F_{1}\right)\boldsymbol{\varepsilon}^{\text{T}}_{q}\boldsymbol{\varepsilon}_{q}\le \frac{k_{1}}{2}F_{1}\ln\left[\cosh\left(\boldsymbol{\varepsilon}^{\text{T}}_{q}\boldsymbol{\varepsilon}_{q}/F_{1}\right)\right]
\end{equation}
 Sorting out these results, we can yield that:
 \begin{equation}\label{dV1}
 	\begin{aligned}
 		\dot{V}_{1} &\le -\left[|q_{e0}|M_{\omega}-4\left(\frac{|\dot{\rho}_{i}|}{\rho_{i}}\right)_{\text{max}}\right]V_{1} + \boldsymbol{P}^{\text{T}}_{q}\boldsymbol{z}_{2}
 	\end{aligned}
 \end{equation}
where  $\boldsymbol{P}_{q} = k_{1}\tanh\left(\boldsymbol{\varepsilon}^{\text{T}}_{q}\boldsymbol{\varepsilon}_{q}/F_{1}\right)\boldsymbol{\psi}_{q}\boldsymbol{F}_{s}\boldsymbol{\varepsilon}_{q}$ is defined accordingly.
Notably, $|q_{e0}|M_{\omega}-4\left(|\dot{\rho}_{i}|/\rho_{i}\right)_{\text{max}} > 0$ should be ensured for the system's stability. This parameter selecting principle will be later discussed in section \ref{para}. 
Subsequently, we consider the range of $\|\boldsymbol{P}_{q}\|$, taking its norm we can yield:
\begin{equation}\begin{aligned}\label{boundq}
			\|\boldsymbol{P}_{q}\| &= \|k_{1}\tanh\left(\frac{\boldsymbol{\varepsilon}^{\text{T}}_{q}\boldsymbol{\varepsilon}_{q}}{F_{1}}\right)\boldsymbol{\psi}_{q}\boldsymbol{F}_{s}\boldsymbol{\varepsilon}_{q}\|\\
			&\le k_{1}\|\boldsymbol{\varepsilon}_{q}\|\|\boldsymbol{\psi}_{q}\|\|\boldsymbol{F}_{s}\|
	\end{aligned}
\end{equation}
According to the conclusion analyzed in remark \ref{boundedeps}, each $\varepsilon_{i}$ can be regarded practically bounded by $k_{u}$, hence $\|\boldsymbol{\varepsilon}_{q}\| \le \sqrt{3}k_{u}$ will be satisfied. Therefore, $\|\boldsymbol{P}_{q}\|$ will be practically bounded such that $\|\boldsymbol{P}_{q}\| \le \sqrt{3}k_{u}k_{1}\left(1/\left(\rho_{i}\right)_{\text{min}}\right)$ will be hold all along. We further define the expressed term  $\sqrt{3}k_{1}k_{u}\left(1/\left(\rho_{i}\right)_{\text{min}}\right)$ as $B_{q}$ for convenient.

	Further, we consider the second-layer error subsystem $\boldsymbol{z}_{2}$ modeled with sample measurement error $\boldsymbol{e}_{u}$.
Substituting the actual control law (\ref{controllaw}) into equation (\ref{systemtrans}), one has:
\begin{equation}\label{subJz}
	\boldsymbol{J}\dot{\boldsymbol{z}}_{2} = -K_{2}\boldsymbol{z}_{2} + \tilde{\boldsymbol{d}} - \boldsymbol{e}_{u}-P_{q}
\end{equation}
Choose the Lyapunov function for $\boldsymbol{z}_{2}$ error subsystem as $V_{2} = \frac{1}{2}\boldsymbol{z}^{\text{T}}_{2}\boldsymbol{J}\boldsymbol{z}_{2}$, taking the time-derivative of $V_{2}$, one can be obtained that:
\begin{equation}\label{dV2}
	\dot{V}_{2} = \boldsymbol{z}^{\text{T}}_{2}\boldsymbol{J}\dot{\boldsymbol{z}}_{2}
\end{equation}
Further, substituting the equation (\ref{subJz}) into equation (\ref{dV2}) and applying the Peter-Paul's inequality(Young's equality), we have:
\begin{equation}\label{dv2con}
	\begin{aligned}
		\dot{V}_{2} &= -K_{2}\|\boldsymbol{z}_{2}\|^{2} + \boldsymbol{z}_{2}^{\text{T}}\tilde{\boldsymbol{d}} - \boldsymbol{z}^{\text{T}}_{2}\boldsymbol{e}_{u} - \boldsymbol{z}^{\text{T}}_{2}\boldsymbol{P}_{q}\\
		&\le -\left(K_{2}-\frac{b}{2}\right)\|\boldsymbol{z}_{2}\|^{2} + \frac{1}{2b}\|\boldsymbol{e}_{u}\|^{2}+\boldsymbol{z}^{\text{T}}_{2}\tilde{\boldsymbol{d}}-\boldsymbol{z}^{\text{T}}_{2}\boldsymbol{P}_{q}
	\end{aligned}
\end{equation}
Noticing the term expressed as  $\boldsymbol{z}_{2}^{\text{T}}\tilde{\boldsymbol{d}}$, by applying the Lemma [REF] mentioned in [REF], we have:
\begin{equation}
	\begin{aligned}
		\boldsymbol{z}^{\text{T}}_{2}\tilde{\boldsymbol{d}} &= \sum_{i=1}^{3}z_{2i}\left[d_{i}-D_{m}\tanh\left(\frac{z_{2i}}{p_{i}}\right)\right]\\
		& \le  \sum_{i=1}^{3}D_{m}\left[|z_{2i}| - z_{2i} \tanh\left(\frac{z_{2i}}{p_{i}}\right)\right]\\
		& \le \sum_{i=1}^{3}0.2785p_{i}D_{m}
	\end{aligned}
\end{equation}
For convenient, define $p_{1} = p_{2} = p_{3} = p$, hence $\boldsymbol{z}^{\text{T}}_{2}\tilde{\boldsymbol{d}}$ will be bounded by $0.2785\cdot3pD_{m} = 0.8355pD_{m}$. Similarly, we define $D_{0} = 0.8355pD_{m}$ for convenient. Notably, for $\|\boldsymbol{z}_{2}\|^{2}$, a vital relationship should be noticed that $\frac{1}{2}\max\left(\lambda_{\boldsymbol{J}}\right)\|\boldsymbol{z}_{2}\|^{2} \ge V_{2}$ will be satisfied, thus we can yield:
\begin{equation}\label{dotV2}
	\dot{V}_{2}\le-2\left(K_{2}-\frac{b}{2}\right)\frac{V_{2}}{\left(\lambda_{\boldsymbol{J}}\right)_{\text{max}}}+\frac{1}{2b}\|\boldsymbol{e}_{u}\|^{2} + D_{0} - \boldsymbol{z}^{\text{T}}_{2}\boldsymbol{P}_{q}
\end{equation}

\subsection{System Stability with Actuator On}
\begin{theorem}
	\label{T1}
	For the system defined in equation (\ref{errorsystem}), under the given virtual control law $\boldsymbol{\alpha}$ in equation (\ref{virtual}), the given actual control law $\boldsymbol{u}$ in equation (\ref{controllaw}) and the designed event trigger mechanism (\ref{trigoff}) and (\ref{trigon}), the system will be exponentially converged.
\end{theorem}
\begin{proof}
	Defining a lumped Lyapunov function $V\left(t\right) = V_{1}\left(t\right) + V_{2}\left(t\right)$ for further analysis. We will omit the time variable $t$ for brevity occasionally.
	Taking the time-derivative of $V$ and combined with the result in equation (\ref{dV1}) and (\ref{dotV2}), one can be yielded that:
	\begin{equation}\begin{aligned}
			\dot{V} &= \dot{V}_{1} + \dot{V}_{2}\\
			&\le -B_{1}V_{1} - B_{2}V_{2} + \frac{1}{2b}\|\boldsymbol{e}_{u}\|^{2} + D_{0}
		\end{aligned}
	\end{equation}
	where $B_{1}$ is defined as $B_{1} = |q_{e0}|M_{\omega} - 4\left(\dot{\rho}_{i}/\rho_{i}\right)_{\text{max}}$ and $B_{2}$ is defined as $B_{2} = 2\left(K_{2}-b/2\right)/\left(\lambda_{\boldsymbol{J}}\right)_{\text{max}}$. A main principle for parameter selecting is that we should ensure $B_{1} > 0$ and $B_{2} > 0$ can be satisfied all along. Accordingly, it can be derived that:
	\begin{equation}
		\dot{V} \le -\left(B_{1},B_{2}\right)_{\text{min}}V + \frac{1}{2b}\|\boldsymbol{e}_{u}\|^{2} + D_{0}
	\end{equation}
	Further, we define $C_{1} = \left(B_{1},B_{2}\right)_{\text{min}}$ for convenient.
	Considering the aforementioned trigger condition depicted in equation (\ref{trigoff}), $\|\boldsymbol{e}_{u}\|^{2} \le s e^{-\beta t} + m$ will be hold during the control procedure. Hence, we can further obtain:
	\begin{equation}
		\dot{V}\left(t\right) \le -C_{1}V\left(t\right) + \frac{1}{2b}\left(se^{-\beta t}+m\right) + D_{0}
	\end{equation}
	By applying the comparison lemma stated in \cite{khalil2015nonlinear}, the Lyapunov trajectory $V\left(t\right) \left(t\in\left[t^{\text{on}}_{k} , t^{\text{off}}_{k}\right]\right)$ will remain under a curve $P_{v}\left(t\right)$ defined as the solution of:
	\begin{equation}\label{diff}
		\dot{P}_{v}\left(t\right)= -C_{1}P_{v}\left(t\right)
		+\frac{1}{2b}se^{-\beta t} + D_{0} + \frac{m}{2b}
	\end{equation} where the initial value is defined as $P_{v}\left(t^{\text{on}}_{k}\right) = V\left(t^{\text{on}}_{k}\right)$. Correspondingly, integrating the differential equation (\ref{diff}) from both side for $t\in\left[t^{\text{on}}_{k},t^{\text{off}}_{k}\right]$, one can be obtained that:
	\begin{equation}\begin{aligned}
			V\left(t\right) &\le V\left(t^{\text{on}}_{k}\right)e^{-C_{1}\left(t-t^{\text{on}}_{k}\right)} \\ 
			&\quad + \frac{1}{2b}\int_{t^{\text{on}}_{k}}^{t}e^{-C_{1}\left(t - \tau\right)}\cdot se^{-\beta \tau}d\tau\\
			&\quad + \left(D_{0}+\frac{m}{2b}\right)\int_{t^{\text{on}}_{k}}^{t}e^{-C_{1}\left(t - \tau\right)}d\tau
		\end{aligned}
	\end{equation}
	Notably, we can rearrange the solution into a standard exponential function. We further define $ V_{0} = \left[V\left(t^{\text{on}}_{k}\right)-\frac{se^{-\beta t^{\text{on}}_{k}}}{2b\left(C_{1}-\beta\right)}-\frac{1}{C_{1}}\left(D_{0}+\frac{m}{2b}\right)\right]$ for brevity, hence we have the following result, depicted as follows:
	\begin{equation}
		\label{Vtraj}
		\begin{aligned}
			V\left(t\right) &\le V\left(t^{\text{on}}_{k}\right)e^{-C_{1}\left(t-t^{\text{on}}_{k}\right)}\\
			& \quad + \frac{se^{-\beta t}}{2b\left(C_{1}-\beta\right)}-\frac{se^{-\beta t^{\text{on}}_{k}}}{2b\left(C_{1}-\beta\right)}e^{-C_{1}\left(t-t^{\text{on}}_{k}\right)}\\
			& \quad+ \frac{1}{C_{1}}\left(D_{0}+\frac{m}{2b}\right)\\
			&\quad-\frac{1}{C_{1}}\left(D_{0}+\frac{m}{2b}\right)e^{-C_{1}\left(t-t^{\text{on}}_{k}\right)}\\
			&=V_{0}e^{-C_{1}\left(t-t^{\text{on}}_{k}\right)}+\frac{se^{-\beta t}}{2b\left(C_{1}-\beta\right)}\\
			&\quad +\frac{1}{C_{1}}\left(D_{0}+\frac{m}{2b}\right)
		\end{aligned}
	\end{equation}
	The parameter designing should ensure that $V_{0}>0$, $C_{1}-\beta >0$ is hold. Therefore, for arbitrary $t\in\left(0,+\infty\right)$, $e^{-C_{1}t}<e^{-\beta t}$ will be always hold. Substituting this relationship into equation (\ref{Vtraj}) and rearrange the equation, it can be further obtained that:
	\begin{equation}\label{eqn17}
		\begin{aligned}
			V\left(t\right) &\le  V_{0}e^{-\beta\left(t-t^{\text{on}}_{k}\right)} + \frac{se^{-\beta t^{\text{on}}_{k}}e^{-\beta\left(t-t^{\text{on}_{k}}\right)}}{2b\left(C_{1}-\beta\right)}\\
			& +\frac{1}{C_{1}}\left(D_{0}+\frac{m}{2b}\right)\\
			& = \left(V\left(t^{\text{on}}_{k}\right)-\frac{1}{C_{1}}\left(D_{0}+\frac{m}{2b}\right)\right)e^{-\beta\left(t-t^{\text{on}}_{k}\right)}\\
			&+\frac{1}{C_{1}}\left(D_{0}+\frac{m}{2b}\right)
		\end{aligned}
	\end{equation}
	Define $V_{\infty} = \frac{1}{C_{1}}\left(D_{0}+\frac{m}{2b}\right)$, the result in equation (\ref{eqn17}) can be rewritten as:
	\begin{equation}\label{V2CONC}
		V\left(t\right) \le \left(V\left(t^{\text{on}}_{k}\right)-V_{\infty}\right)e^{-\beta\left(t-t^{\text{on}}_{k}\right)}+V_{\infty}
	\end{equation}
	The right hand side of the inequality (\ref{V2CONC}) is an exponential function, starting from $t = t^{\text{on}}_{k}$. The initial value is $V\left(t^{\text{on}}_{k}\right)$ and the terminal value is tend to be $V_{\infty}$.
	Such a result indicates that under the given control law $\boldsymbol{u}$ stated in equation (\ref{controllaw}) and the event trigger mechanism (\ref{trigoff}), the whole system will be exponentially converged during the control time interval that $t\in\left[t^{\text{on}}_{k},t^{\text{off}}_{k}\right]$. It can be concluded that during the control procedure, the Lyapunov trajectory $V\left(t\right)$ will remain under an exponential function. This ensures that the system will strictly converged during the control procedure, ensuring the system's stability for $t\in\left[t^{\text{on}}_{k},t^{\text{off}}_{k}\right]$. This completes the proof of theorem (\ref{T1}). 
\end{proof}

\subsection{System Behavior with Actuator Off}\label{secAoff}
 In this subsection, we discuss the system behavior with actuator off. Define the corresponding performance evaluation function of $V_{2}\left(t\right)$ as $S_{2}\left(t\right)$, it can be noticed that $S_{2}\left(t\right) > V_{2}\left(t\right)$ will be satisfied all along due to the trigger mechanism. However, to further derive the conclusion about the trigger time, the upper boundary for $V_{2}\left(t\right)$ during the non-control time interval is investigated in this subsection.

\begin{theorem}\label{T2}
	For the second-layer error system $\boldsymbol{z}_{2}$ with none control input, i.e. $\boldsymbol{u}\left(t\right)=\boldsymbol{0}$, the system trajectory $V_{2}\left(t\right)$ will still remain beneath a specific function dentoed as $P_{l}\left(t\right)$, specified in equation (\ref{V2bound}).
\end{theorem}

\begin{proof}
	When the controller is turned off such that $\boldsymbol{u}\left(t\right) = \boldsymbol{0}$ is hold, one can be obtained that:
	\begin{equation}
		\boldsymbol{J}\dot{\boldsymbol{z}}_{2} = -\boldsymbol{\omega}^{\times}_{s}\boldsymbol{J}\boldsymbol{\omega}_{s} + \boldsymbol{u}\left(t\right)- \boldsymbol{e}_{u} +  \boldsymbol{d} - \boldsymbol{J}\dot{\boldsymbol{\alpha}}
	\end{equation}
	Notably, $\boldsymbol{u}\left(t\right)-\boldsymbol{e}_{u}  = \boldsymbol{0}$ will be hold under current condition. Accordingly, considering the time-derivative of $V_{2}$, we have:
	\begin{equation}
		\begin{aligned}
			\dot{V}_{2} = \boldsymbol{z}^{\text{T}}_{2}\boldsymbol{J}\dot{\boldsymbol{z}}_{2} \le \|\boldsymbol{z}_{2}\|\|\boldsymbol{J}\dot{\boldsymbol{z}}_{2}\|
		\end{aligned}
	\end{equation}
	In conclusion, it should be noticed that $\|-\boldsymbol{\omega}^{\times}_{s}\boldsymbol{J}\boldsymbol{\omega}_{s}\| \le \|\boldsymbol{J}\|\|\boldsymbol{\omega}_{s}\|^{2}\le \left(\lambda_{\boldsymbol{J}}\right)_{\text{max}}B^{2}_{\omega}$ is hold, where $B_{\omega}$ is the upper boundary for angular velocity constraint defined in remark [\ref{boundedv}]. Further, $\|\dot{\boldsymbol{\alpha}}\| \le B_{\alpha}$ is hold according to remark (\ref{boundedAlpha}). Combining these results, this derives the following result:
	\begin{equation}
		\|\boldsymbol{J}\dot{\boldsymbol{z}}_{2}\|\le\left(\lambda_{\boldsymbol{J}}\right)_{\text{max}}B^{2}_{\omega}+D_{m}+\left(\lambda_{\boldsymbol{J}}\right)_{\text{max}}B_{\alpha}
	\end{equation}
	We further define the upper boundary for  $\|\boldsymbol{J}\dot{\boldsymbol{z}}_{2}\|$ as $U_{z} = \left(\lambda_{\boldsymbol{J}}\right)_{\text{max}}B^{2}_{\omega}+D_{m}+\left(\lambda_{\boldsymbol{J}}\right)_{\text{max}}B_{\alpha}$. Correspondingly, $\dot{V}_{2}$ satisfies the following relationship:
	\begin{equation}
		\dot{V}_{2} \le U_{z}\|\boldsymbol{z}_{2}\|
	\end{equation}
According to the comparison lemma \cite{khalil2015nonlinear}, the $V_{2}\left(t\right)$ trajectory will remain beneath a function $P_{l}$, defined as the solution of the following differential equation, expressed as:
\begin{equation}\label{diffz2}
	\begin{aligned}
		\dot{P}_{l} &= U_{z}\sqrt{\frac{2}{\left(\lambda_{\boldsymbol{J}}\right)_{\text{min}}}}P^{1/2}_{l}
	\end{aligned}
\end{equation}
where the initial condition should satisfy $ P_{l}\left(t^{\text{off}}_{k}\right) = V_{2}\left(t^{\text{off}}_{k}\right)$. For $t\in\left[t^{\text{off}}_{k},t^{\text{on}}_{k+1}\right]$, integrating the equation (\ref{diffz2}) from both side, we have: 
\begin{equation}\label{V2bound}
	V_{2}\left(t\right)\le P_{l}\left(t\right) = \frac{a^{2}_{0}}{4}\left(t - t^{\text{off}}_{k} + C_{m}\right)^{2}
\end{equation}
where $a_{0}$ is defined as $a_{0}=U_{z}\sqrt{2/\left(\lambda_{\boldsymbol{J}}\right)_{\text{min}}}$.
Considering the initial condition of $P_{l}\left(t\right)$, we have $a^{2}_{0}C^{2}_{m}/4 = V_{2}\left(t^{\text{off}}_{k}\right)$, hence it derives $C_{m} = \frac{2}{a_{0}}\sqrt{V_{2}\left(t^{\text{off}}_{k}\right)}$.
Equivalently, equation (\ref{V2bound}) can be rewritten as:
\begin{equation}\begin{aligned}\label{Vbound2}
		V_{2}\left(t\right)
		\le V_{2}\left(t^{\text{off}}_{k}\right) + \frac{a^{2}_{0}}{4}\left(t-t^{\text{off}}_{k}\right)^{2} + \frac{a^{2}_{0}}{2}\left(t-t^{\text{off}}_{k}\right)C_{m}\\
	\end{aligned}
\end{equation} 

As a result, we have the upper boundary of the system state $V_{2}\left(t\right)$ when there is no control input, which is actually a quadratic function. This completes the proof of theorem [\ref{T2}].

Further, we discuss the stability of the whole system when there is no control input. Although the actuator is \textbf{turned off}, however, the trigger mechanism guarantees that $S_{2} > V_{2}$ will be always hold. Here we proof that the output-layer system Lyapunov trajectory $V_{1}\left(t\right)$ is still beneath a exponential-converged function.
\begin{theorem}
	For the first-layer error subsystem $\boldsymbol{\varepsilon}_{q}$, under the aforementioned controller, $V_{1}$ can be still exponentially converged under the current design.
\end{theorem}
It can be noticed that the second-layer error subsystem will ensure that $V_{2}\left(t\right) < S_{2}\left(t\right)$ will be always satisfied. Hence we have:
\begin{equation}
	\frac{1}{2}\left(\lambda_{\boldsymbol{J}}\right)_{\text{min}}\|\boldsymbol{z}_{2}\|^{2} \le \frac{1}{2}\boldsymbol{z}^{\text{T}}_{2}\boldsymbol{J}\boldsymbol{z}_{2} \le S_{2}\left(t\right)
\end{equation}
Such a result can further derive that $\|\boldsymbol{z}_{2}\|^{2} \le \frac{2S_{2}}{\left(\lambda_{\boldsymbol{J}}\right)_{\text{min}}}$. Considering the first-layer error subsystem $\boldsymbol{\varepsilon}_{q}$, one can be obtained from equation (\ref{dV1}) that:
\begin{equation}
	\dot{V}_{1} \le -C_{\varepsilon}V_{1} + \boldsymbol{P}^{\text{T}}_{q}\boldsymbol{z}_{2}
\end{equation}
where $C_{\varepsilon}$ is defined as $C_{\varepsilon} = |q_{e0}M_{\omega}| - 4\left(\dot{\rho}_{i}/\rho_{i}\right)_{\text{max}}$. Applying the Peter-Paul inequality, it can be further derived that:
\begin{equation}\begin{aligned}\label{V1dot}
		\dot{V}_{1} &\le -C_{\varepsilon}V_{1} + \|\boldsymbol{P}_{q}\|\|\boldsymbol{z}_{2}\|\\
		&\le -C_{\varepsilon}V_{1}\\
		&\quad + B_{q}\left[\sqrt{\frac{2}{\left(\lambda_{\boldsymbol{J}}\right)_{\text{min}}}}\left(\sqrt{\left(S_{2}\left(0\right)-S_{2\infty}\right)}e^{-\frac{\beta}{2}t}\right)\right]\\
		&\quad + B_{q}\sqrt{\frac{2}{\left(\lambda_{\boldsymbol{J}}\right)_{\text{min}}}}\sqrt{S_{2\infty}}
	\end{aligned}
\end{equation}
Define $G_{s} = B_{q}\sqrt{\frac{2}{\left(\lambda_{\boldsymbol{J}}\right)_{\text{min}}}}\sqrt{\left(S_{2}\left(0\right)-S_{2\infty}\right)}$ and $G_{m} = B_{q}\sqrt{\frac{2}{\left(\lambda_{\boldsymbol{J}}\right)_{\text{min}}}}\sqrt{S_{2\infty}}$, the inequality (\ref{V1dot}) can be further rearranged as:
\begin{equation}
	\dot{V}_{1} \le -C_{\varepsilon}V_{1} + G_{s}e^{-\frac{\beta}{2} t} + G_{m}
\end{equation}
Such a differential equation is corresponding to the following solution if $C_{\varepsilon} > \frac{1}{2}\beta$ can be satisfied, expressed as follows:
\begin{equation}
	V_{1}\left(t\right) \le \left(V_{1}\left(0\right) - \frac{1}{C_{\varepsilon}}G_{m}\right)e^{-\frac{\beta}{2} t} + \frac{1}{C_{\varepsilon}}G_{m}
\end{equation}
Accordingly, it can be observed when $V_{2}$ remains beneath the exponentially-converged function $S_{2}$, the output layer $\boldsymbol{\varepsilon}$ will be exponentially converged, which completes the proof of the system stability.

Further, it can be proved that the inter-event time between the neighbored two triggering time instant $t^{\text{on}}_{k}$ and $t^{\text{off}}_{k}$ will be globally lower-bounded by positive constant. Similarly, the time interval between $t^{\text{off}}_{k}$ and $t^{\text{on}}_{k+1}$ will also be lower bounded. This indicates that the trigger condition will not be satisfied repeatedly in a tiny time interval that is tend to be zero, which ruling out the possibility of the zeno behavior. 
The detailed proof of the MIET (minimum inter-event time) of the presented composite event-trigger mechanism will be elaborated in Appendix [\ref{mietofff}] and [\ref{mieton}] later. A globally boundary will be given for both $t^{\text{off}}_{k} - t^{\text{on}}_{k}$ and $t^{\text{on}}_{k+1} - t^{\text{off}}_{k}$, which can be analyzed analytically.

\subsection{Suggestions for System Parameter Selecting}\label{para}
In this section, we give some suggestions about the main principle for parameter selecting of the proposed scheme.

\textbf{1.} For the stable of the first-layer error subsystem $\boldsymbol{\varepsilon}_{q}$, according to the equation (\ref{dV1}), $|q_{e0}M_{\omega}| - 4\left(|\dot{\rho}_{i}|/\rho_{i}\right)_{\text{max}} > 0$ should be satisfied for the exponential converge of $V_{1}$.

\textbf{2.} According to the analysis in subsection [\ref{sectrigon}], $C_{1} < \beta$ and $C_{\varepsilon} < \beta/2$ should be satisfied. This indicate that the design of $\beta$ should not be beyond the maximum ability of the system convergence capability.
Moreover, $\beta$ is also related to $S_{2}\left(t\right)$. As stated before, some fundamental requirements should be satisfied, such as $S_{2}\left(0\right) \ge V_{2}\left(0\right)$ and $S_{2\infty} > V_{2\infty}$. Here we focus on the choosing of $\beta$. 

Suppose $\beta$ is set too big. In that case, the performance envelope will converge rapidly, and $S_{2}\left(t\right)$ may already converge to the neighborhood of $S_{2\infty}$ even when the attitude maneuver has not finished. Since the spacecraft rotates at relatively big angular rates during the attitude maneuver (i.e., $\|\boldsymbol{\omega}_{s}\|$ is not close enough to zero), $V_{2}\left(t\right)$ trajectory will diverge rapidly under such a condition, and the trigger condition (\ref{trigon}) may be satisfied frequently, leading to too-frequent acting of the actuator. To avoid this phenomenon, $\beta$ should be set to a smaller value. A balance should be considered between the convergence rate of the $V_{2}\left(t\right)$ and the triggering frequency.

\textbf{3.} The design of the virtual control law $\boldsymbol{\alpha}$ should not be too aggressive. We suggest that the parameter should make the time-derivative of $\boldsymbol{\alpha}$ relatively small, which can be realized by choosing a smaller $k_{1}$. Intuitively, this can be explained as follows: Under the event-trigger mechanism, the angular velocity is a kind of stair-like signal generated by a zero-order holding sampler. Suppose such a stair signal is expected to track the smooth virtual control law $\boldsymbol{\alpha}$ changes rapidly. In that case, a more frequent updating of the zero-order sampling is required, which results in the over-frequently trigger of the trigger mechanism.

\end{proof}

\section{SIMULATION AND ANALYSIS}\label{sim}
In this section, several groups of simulation results are illustrated to validate the effectiveness of the proposed scheme. Firstly, based on an assumed attitude reorientation task, we present the proposed method's effect on handling such issues.
\subsection{Simulation Scenario Establishment}
In this section, the spacecraft is assumed to be a rigid-body spacecraft, of which the inertial matrix is assumed to be $\boldsymbol{J} = \left(2.8,2.5,1.9\right)kg\cdot m^{2}$. The maximum output torque is set to be $0.05N\cdot m$, and the allowed maximum continuously-output duration time is $1s$.
Further, we utilize the general periodically-varying external disturbance model, expressed as follows:
\begin{equation}
	\boldsymbol{d} = 
	\begin{bmatrix}
		1e-4 \cdot \left[4\sin\left(3\omega_{\text{dis}}t\right) + 3\cos\left(10\omega_{p}t\right) -2\right]\\ 	
		1e-4\left[-1.5\sin\left(2\omega_{\text{dis}}t\right) + 3\cos\left(5\omega_{\text{dis}}t\right) +2\right]\\ 	
		1e-4\left[3\sin\left(10\omega_{\text{dis}}t\right) - 8\cos\left(4\omega_{\text{dis}}t\right) +2\right]\\ 	
	\end{bmatrix}
\end{equation}
where $\omega_{\text{dis}}$ represents the angular frequency of the disturbance, set to be $\omega_{\text{dis}} = 0.01\text{rad}/s$. Such a disturbance is much bigger than the actual one in real aerospace environment, which is enough for validate the proposed scheme.
Subsequently, the angular rate of the spacecraft should not exceed $3^{\circ}/s$ during the attitude adjustment procedure, which is equivalent to $\|\boldsymbol{\omega}_{s}\| \le 0.035\text{rad}/s$. 
To satisfy the constraint, we set $k_{u} = 1.5$ and $M_{\omega} = 0.0524$.

Further, the performance requirements for the controller is stated as follows: \textbf{1.} The system should converge to $\left(|q_{ei}|\right) \le 1e-3 \left(i = 1,2,3\right)$ in no more than $60s$. \textbf{2.} The terminal control accuracy should be greater than $0.05^{\circ}$.

\subsection{Normal Attitude Reorientation Case}
Based on the aforementioned requirements and constraints, the desired attitude quaternion $\boldsymbol{q}_{d}$ and the desired attitude angular velocity $\boldsymbol{\omega}_{d}$ is chosen as follows, respectively.
\begin{equation}
	\begin{aligned}
		\boldsymbol{q}_{d}\left(0\right) &= \left[0.2,-0.5,-0.5,-0.6782\right]^{\text{T}}\\
		\boldsymbol{\omega}_{d}\left(t\right) &= \boldsymbol{0}
	\end{aligned}
\end{equation} 
while the initial condition of the spacecraft is expressed as follows:
\begin{equation}
	\begin{aligned}
	\boldsymbol{q}_{s}\left(0\right) &= \left[0.5175,0.3881,0.4200,0.6366\right]^{\text{T}}\\
	\boldsymbol{\omega}_{s}\left(0\right) &= \boldsymbol{0}
\end{aligned}
\end{equation}
The simulation results are illustrated in figure [\ref{fig_normaQE}][\ref{fig_normaWS}][\ref{fig_normaWSNorm}][\ref{fig_normaU}]. The time-evolution of $q_{ei}\left(t\right)\left(i=1,2,3\right)$ trajectory is illustrated in figure [\ref{fig_normaQE}], while the corresponding $\boldsymbol{\omega}_{s}$ and $\|\boldsymbol{\omega}_{s}\|$ trajectory is presented in figure [\ref{fig_normaWS}] and [\ref{fig_normaWSNorm}] respectively. The actuator output is illustrated in figure [\ref{fig_normaU}], with each component illustrated separately.

\begin{figure}[hbt!]
	\centering 
	\includegraphics[scale = 0.3]{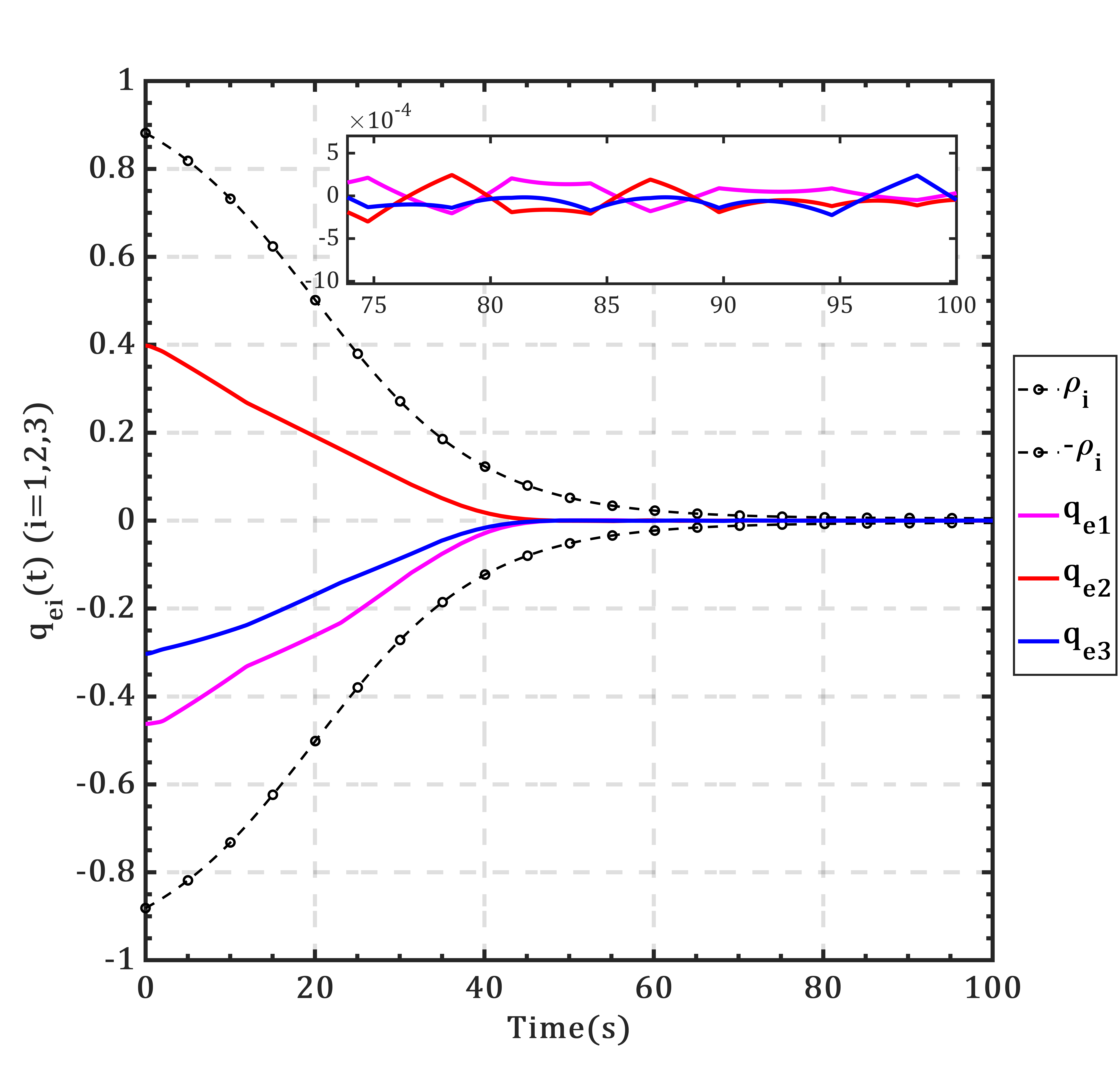}
	\caption{Time Response of attitude error quaternion $q_{ei}\left(t\right) \left(i=1,2,3\right)$}    
	\label{fig_normaQE}  
		\centering 
	\includegraphics[scale = 0.3]{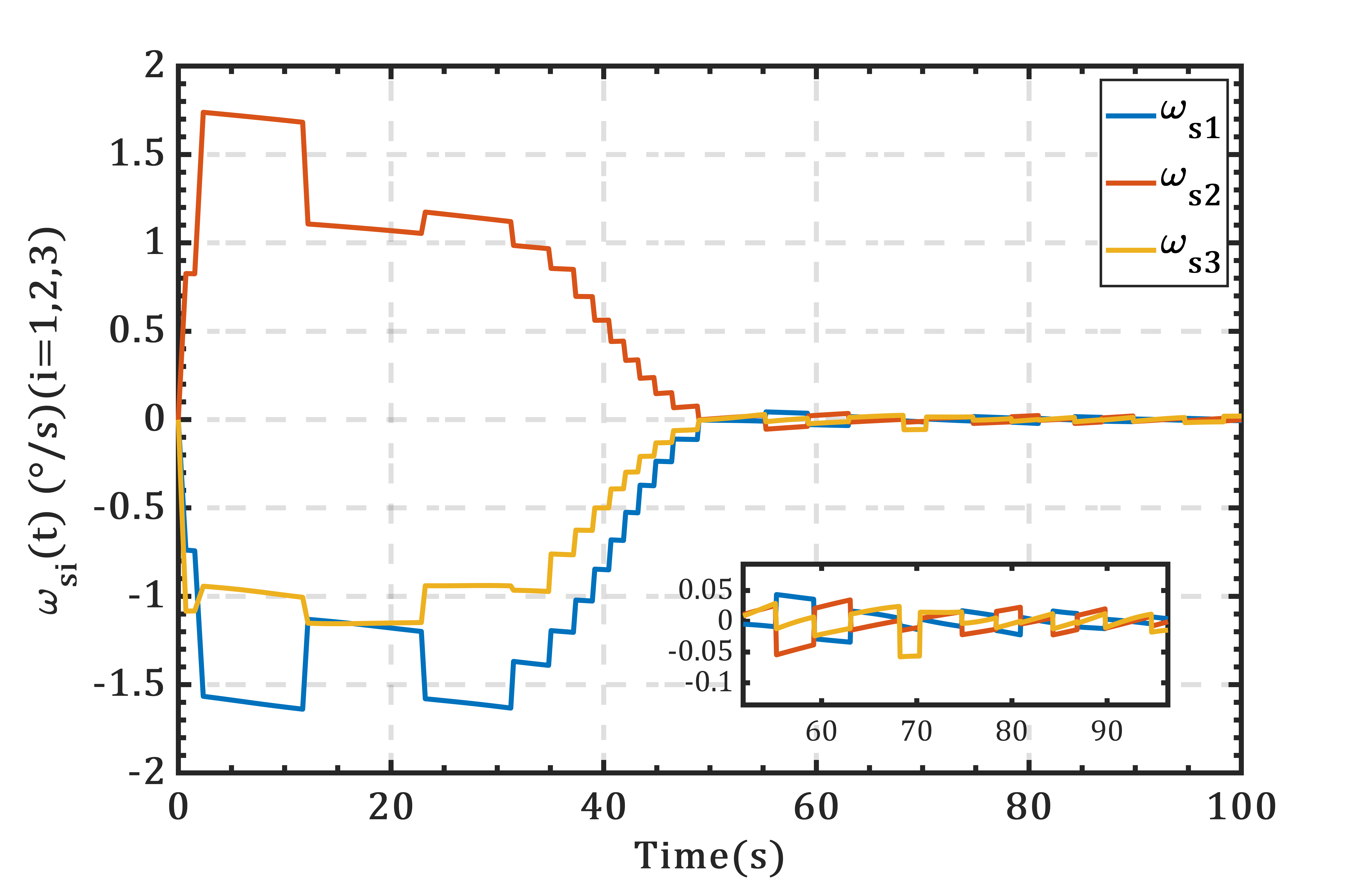}
	\caption{Time Response of attitude angular velocity $\omega_{si}\left(t\right) \left(i=1,2,3\right)$, corresponding to $x,y,z$ axis}    
	\label{fig_normaWS}  
	\centering 
	\includegraphics[scale = 0.3]{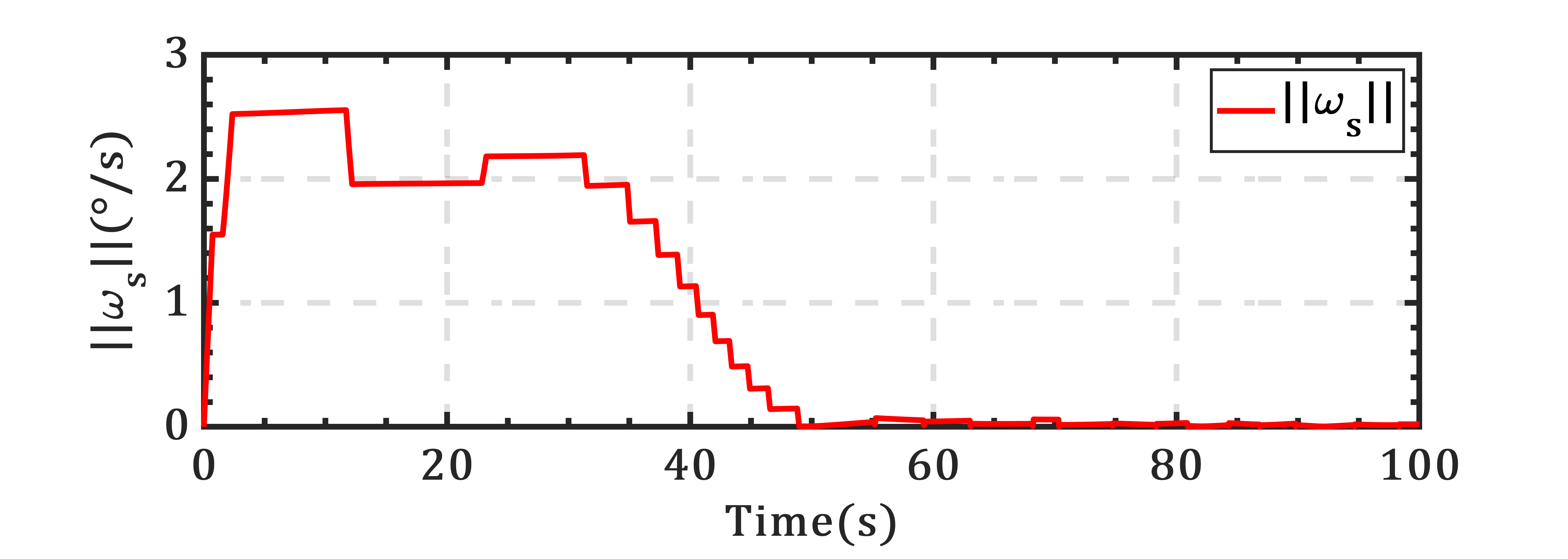}
	\caption{Time Response of the norm of the attitude angular velocity $\|\boldsymbol{\omega}_{s}\|$}    
	\label{fig_normaWSNorm}  
\end{figure}

\begin{figure}
			\centering 
	\includegraphics[scale = 0.32]{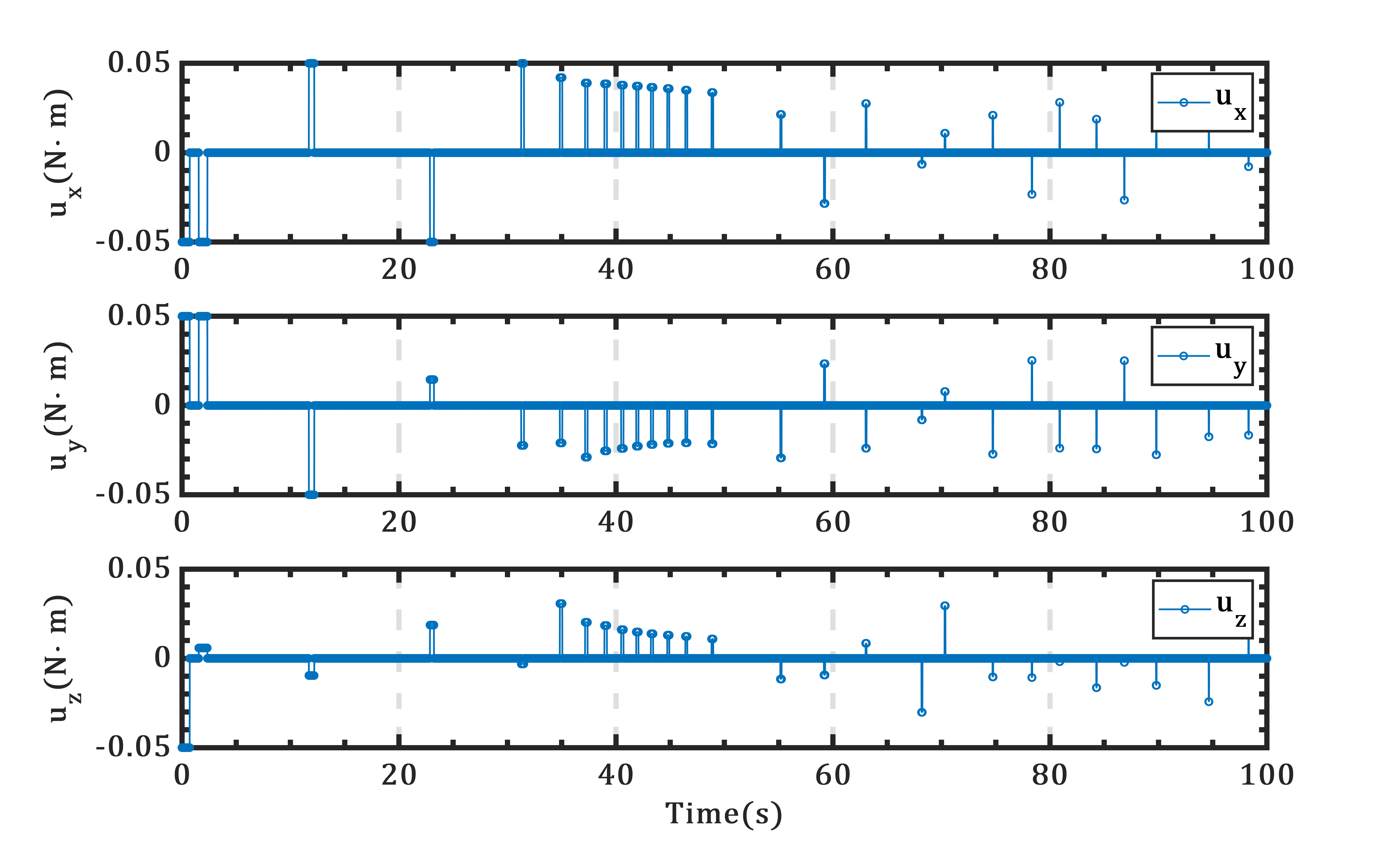}
	\caption{Time Response of the actuator output $\boldsymbol{u}$}    
	\label{fig_normaU}  
				\centering 
	\includegraphics[scale = 0.32]{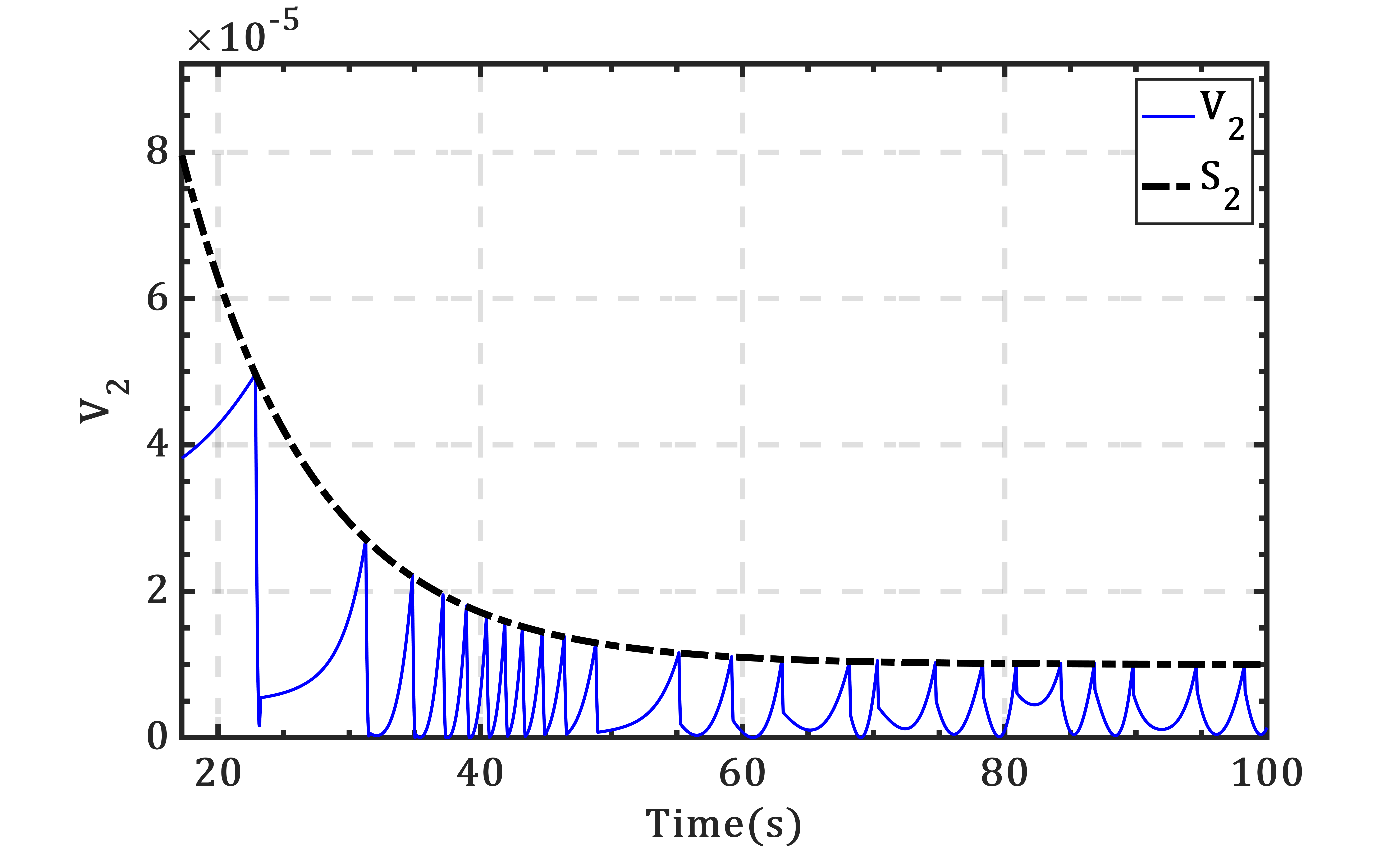}
	\caption{Time Response of $V_{2}$ and $S_{2}$}    
	\label{fig_normaS2}  
\end{figure}

As defined in figure [\ref{fig_normaQE}], the $i$ th component of the attitude error quaternion $\boldsymbol{q}_{e}$ is expressed as $q_{ei}\left(t\right)\left(i = 1,2,3\right)$.
In figure [\ref{fig_normaQE}], it can be observed that the state trajectory $q_{ei}\left(t\right)$ remains in the sector-like constraint region all along. The control accuracy is greater than $0.04^{\circ}/s$, as presented in the scaled part in figure [\ref{fig_normaQE}]. As a result, all the performance requirements are able to be satisfied. Further, it can be observed in figure [\ref{fig_normaU}] that the actuator of the control procedure is working impulsively, which indicates that the actuator acting frequency is significantly reduced. The control frequency for the actuator on each axis is almost $1\text{Hz}$, which is practically realizable in real engineering. Moreover, the angular rotation rate is strictly limited under the given boundary, as shown in figure [\ref{fig_normaWS}] and [\ref{fig_normaWSNorm}].

As an additional statement, the trajectory of $S_{2}\left(t\right)$ and $V_{2}\left(t\right)$ is illustrated in figure [\ref{fig_normaS2}]. It can be observed that the $V_{2}\left(t\right)$ trajectory is like a wavy-curve and remains beneath $S_{2}\left(t\right)$ permanently.

Notably, it can be observed that the trigger condition becomes relatively frequently at around $40s$. This is mainly due to the third issue we mentioned in subsection [\ref{para}]. Since the spacecraft need braking at current time, the virtual control law decreases at a relatively high rate, which causes the increasing of the acting frequency.

The simulation result shows that the proposed scheme is effective and is able to realize the desired control. The actuator is working impulsively during the whole control procedure, conspicuously reducing the work frequency of the propulsion system.

\section{CONCLUSION}
In order to enhance the potential application value of prescribed performance control in attitude control problems, this paper organically combines the intermittent control-based method and the prescribed performance control scheme together. A composite event-trigger mechanism is proposed for this issue, which is alternatively governed by two independent trigger mechanisms. Under the presented trigger mechanism, the actuator will be shut down when there is not necessary to hold a controller output, and the actuator's fuel consumption and acting frequency is conspicuously reduced. Further, specific suggestions for parameter selection are also provided, facilitating its application to various control tasks. As a result, the desired performance is achieved with low-frequency control while the angular velocity constraint is satisfied. It can be observed that through fewer times of control, the attitude deviation is converged to the steady state, with all the performance requirements satisfied. Such a simple controller structure is potentially applicable to the actual on-orbit attitude control scenario. The extensibility of the proposed scheme allows it to be combined with various controller structures, which is worth further investigation.
\section*{APPENDIX}
\subsection{MIET of Turn-Off Trigger Mechanism}\label{mietofff}
In this subsection, we further discuss the time interval between two neighbored $k$ th triggering instant $t^{\text{on}}_{k}$ and $t^{\text{off}}_{k}$, i.e., the minimum inter-event time (MIET) for the turn-off trigger mechanism (\ref{trigoff}). 
In this section, it is elaborated that the MIET time for the trigger mechanism (\ref{trigoff}) is lower-bounded, which is strongly related to the parameter and can be estimated analytically.
\begin{theorem}\label{T3}
	The Minimum Inter-event time (MIET) between $t^{\text{on}}_{k}$ and $t^{\text{off}}_{k}$ will be bounded, and the MIET will not only exists, but strongly related to the system parameters.
\end{theorem}
\begin{proof}
	In order to analyze the time interval of the trigger condition, we start from considering the time-response of $\|\boldsymbol{e}_{u}\|^{2}$.
	For the time-derivative of $\boldsymbol{e}^{\text{T}}_{u}\boldsymbol{e}_{u}$, we have $d\left(\|\boldsymbol{e}_{u}\|^{2}\right)/dt = 2\boldsymbol{e}^{\text{T}}_{u}\dot{\boldsymbol{e}}_{u}$.
	Taking the time-derivative of $\boldsymbol{e}_{u}$, one has:
	\begin{equation}
		\begin{aligned}
			\dot{\boldsymbol{e}_{u}} = \dot{\boldsymbol{u}} &= -\dot{\boldsymbol{\omega}}^{\times}_{s}\boldsymbol{J}\boldsymbol{\omega}_{s} - \boldsymbol{\omega}^{\times}_{s}\boldsymbol{J}\dot{\boldsymbol{\omega}}_{s} -K_{2}\left(\dot{\boldsymbol{\omega}}_{s}-\dot{\boldsymbol{\alpha}}\right)\\
			& \quad + \boldsymbol{J}\ddot{\boldsymbol{\alpha}} - \dot{\hat{\boldsymbol{d}}} - \dot{\boldsymbol{P}}_{q}
		\end{aligned}
	\end{equation}
	Further, we have $\dot{\hat{\boldsymbol{d}}} = \frac{D_{m}}{p}\left[\left(1-\tanh^{2}\left(z_{2i}/p\right)\right)\dot{z}_{2i}\right]_{\text{v}}$. Therefore, it can be concluded that:
	\begin{equation}
		\|\dot{\hat{\boldsymbol{d}}}\|\le\frac{D_{m}}{p}\|\dot{\boldsymbol{z}}_{2}\| = \frac{D_{m}}{p}\left[\|\dot{\boldsymbol{\omega}}_{s}\| + \|\dot{\boldsymbol{\alpha}}\|\right]
	\end{equation}
	Meanwhile, according to the aforementioned analysis, we have $\|\dot{\boldsymbol{\alpha}}\| \le B_{\alpha}$ and $\|\ddot{\boldsymbol{\alpha}}\|\le B_{2\alpha}$, as analyzed in remark \ref{boundedAlpha}.
	Sorting out these conclusions, $\|\dot{\boldsymbol{u}}\|$ can further considered as follows:
	\begin{equation}\begin{aligned}\label{dotU}
			\|\dot{\boldsymbol{u}}\|&\le \left(\lambda_{\boldsymbol{J}}\right)_{\text{max}}B_{\omega}\|\dot{\boldsymbol{\omega}}_{s}\|+B_{\omega}\|\boldsymbol{J}\dot{\boldsymbol{\omega}}_{s}\|\\
			&\quad+\left(K_{2}+\frac{D_{m}}{p}\right)\|\dot{\boldsymbol{\omega}}_{s}\| + \left(K_{2}+\frac{D_{m}}{p}\right)\|\dot{\boldsymbol{\alpha}}\|\\
			&\quad +\left(\lambda_{\boldsymbol{J}}\right)_{\text{max}}B_{2\alpha} + \|\dot{\boldsymbol{P}}_{q}\|
		\end{aligned}
	\end{equation}
	
	Next, we discuss about the range of $\|\dot{\boldsymbol{\omega}}_{s}\|$. Since the following relationship holds all along with $\boldsymbol{u}\left(t\right)$:
	\begin{equation}\label{jz2}
		\boldsymbol{J}\dot{\boldsymbol{\omega}}_{s} = -K_{2}\boldsymbol{z}_{2} - \boldsymbol{e}_{u} + \tilde{\boldsymbol{d}} + \boldsymbol{J}\dot{\boldsymbol{\alpha}}-\boldsymbol{P}_{q}
	\end{equation}
	where $\tilde{\boldsymbol{d}} = \boldsymbol{d}-\hat{\boldsymbol{d}}$.
	It can be further obtained that $\|\tilde{\boldsymbol{d}}\| \le \|\boldsymbol{d}\| + \|\hat{\boldsymbol{d}}\| \le \left(1+\sqrt{3}\right)D_{m}$ is hold. Sorting up these conclusion, we can further obtain that:
	\begin{equation}\begin{aligned}\label{JdW}			\|\boldsymbol{J}\dot{\boldsymbol{\omega}}_{s}\| &\le K_{2}\|\boldsymbol{z}_{2}\| + \|\boldsymbol{e}_{u}\| + \left(1+\sqrt{3}\right)D_{m} \\&\quad+ \left(\lambda_{\boldsymbol{J}}\right)_{\text{max}}B_{\alpha} + \sqrt{3}k_{1}k_{u}\frac{1}{\left(\rho_{i}\right)_{\text{min}}}
		\end{aligned}
	\end{equation}
	Notably, $\|\dot{\boldsymbol{\omega}}_{s}\| = \|\boldsymbol{J^{-1}\boldsymbol{J}}\dot{\boldsymbol{\omega}}_{s}\| \le \|\boldsymbol{J}^{-1}\|\|\boldsymbol{J}\dot{\boldsymbol{\omega}}_{s}\|$ can be noticed.
	Define $\left(1+\sqrt{3}\right)D_{m} + \left(\lambda_{\boldsymbol{J}}\right)_{\text{max}}B_{\alpha} + \sqrt{3}k_{1}k_{u}/\left(\rho_{i}\right)_{\text{min}}$ as $B_{c}$, substituting the result in equation (\ref{JdW}) into $\dot{\boldsymbol{u}}$, the upper boundary of $\dot{\boldsymbol{u}}$ can be expressed as follows:
	\begin{equation}
		\begin{aligned}
			\|\dot{\boldsymbol{u}}\| &\le \left(\lambda_{\boldsymbol{J}}\right)_{\text{max}}B_{\omega}\|\boldsymbol{J}^{-1}\|\left(K_{2}\|\boldsymbol{z}_{2}\|+\|\boldsymbol{e}_{u}\|+B_{c}\right) \\
			&\quad + B_{\omega}\left(K_{2}\|\boldsymbol{z}_{2}\|+\|\boldsymbol{e}_{u}\|+B_{c}\right)\\
			&\quad + \left(K_{2}+\frac{D_{m}}{p}\right)\|\boldsymbol{J}^{-1}\|\left(K_{2}\|\boldsymbol{z}_{2}\|+\|\boldsymbol{e}_{u}\|+B_{c}\right) \\
			&\quad + \left(K_{2}+\frac{D_{m}}{p}\right)B_{\alpha}  + \left(\lambda_{\boldsymbol{J}}\right)_{\text{max}}B_{2\alpha} +\|\dot{\boldsymbol{P}}_{q}\|
		\end{aligned}
	\end{equation}
	Further, since $\|\boldsymbol{P}_{q}\|$ is practically bounded by $\sqrt{3}k_{1}k_{u}/\left(\rho_{i}\right)_{\text{min}}$and is a continuously varying variable as we mentioned in equation (\ref{boundq}) , hence it is rational to assume that $\dot{\boldsymbol{P}}_{q}$ will also be bounded by a positive constant, and we define such an upper boundary as $B_{q}$ here.
	Sorting our these result, this can be further rearranged into the following form:
	\begin{equation}\label{eu}
		\|\dot{\boldsymbol{u}}\| = \|\dot{\boldsymbol{e}}_{u}\| \le R_{1}K_{2}\|\boldsymbol{z}_{2}\| + R_{1}\|\boldsymbol{e}_{u}\| + R_{2}
	\end{equation}
	where $R_{1}$, $R_{2}$ are defined for brevity, expressed as:
	\begin{equation}
		\begin{aligned}
			R_{1} &= \left(\lambda_{\boldsymbol{J}}\right)_{\text{max}}B_{\omega}\|\boldsymbol{J}^{-1}\|+B_{\omega}\\
			&\quad+\left(K_{2}+\frac{D_{m}}{p}\right)\|\boldsymbol{J}^{-1}\|\\
			R_{2} &= \left(K_{2}+\frac{D_{m}}{p}\right)B_{\alpha} + \left(\lambda_{\boldsymbol{J}}\right)_{\text{max}}B_{2\alpha} \\
			&\quad+ B_{c}\left(\lambda_{\boldsymbol{J}}\right)_{\text{max}}B_{\omega}\|\boldsymbol{J}^{-1}\|\\
			&\quad +B_{c}B_{\omega}+B_{c}\left(K_{2}+\frac{D_{m}}{p}\right)\|\boldsymbol{J}^{-1}\| + B_{q}
		\end{aligned}
	\end{equation}
	Substituting equation (\ref{eu}) into $d\left(\|\boldsymbol{e}_{u}\|^{2}\right)/dt = 2\boldsymbol{e}^{\text{T}}_{u}\dot{\boldsymbol{e}}_{u}$ and applying the Young's inequality, we have:
	\begin{equation}
		\begin{aligned}
			\boldsymbol{e}^{\text{T}}_{u}\dot{\boldsymbol{e}}_{u} &\le 2\|\boldsymbol{e}_{u}\|\|\dot{\boldsymbol{e}}_{u}\|\\
			& \le 2\|\boldsymbol{e}_{u}\|\left[R_{1}K_{2}\|\boldsymbol{z}_{2}\|+R_{1}\|\boldsymbol{e}_{u}\|+R_{2}\right]\\
			&=2R_{1}K_{2}\|\boldsymbol{e}_{u}\|\|\boldsymbol{z}_{2}\|+2R_{1}\|\boldsymbol{e}_{u}|^{2}+2R_{2}\|\boldsymbol{e}_{u}\|\\
			&\le \left[R_{1}K_{2}+2R_{1}+1\right]\|\boldsymbol{e}_{u}\|^{2} \\
			&\quad + R_{1}K_{2}\|\boldsymbol{z}_{2}\|^{2}+R^{2}_{2}
		\end{aligned}
	\end{equation}
	Similarly, note that $\|\boldsymbol{z}_{2}\|^{2} \le \frac{2}{\left(\lambda_{\boldsymbol{J}}\right)_{\text{max}}}V_{2}$ is satisfied all along.
	Substituting the conclusion in equation (\ref{V2CONC}) and the trigger condition (\ref{trigoff}) into it, we can further obtain:
	\begin{equation}\begin{aligned}
			\boldsymbol{e}^{\text{T}}_{u}\dot{\boldsymbol{e}}_{u} &\le
			Q_{1}\left(se^{-\beta t}+m\right)\\
			&\quad+Q_{2}\left[G_{1}e^{-\beta \left(t-t^{\text{on}}_{k}\right)} + G_{2}\right] + Q_{3}\\
			&\le	Q_{1}\left(se^{-\beta \left(t-t^{\text{on}}_{k}\right)}+m\right)\\
			&\quad+Q_{2}\left[G_{1}e^{-\beta \left(t-t^{\text{on}}_{k}\right)} + G_{2}\right] + Q_{3}
		\end{aligned}
	\end{equation}
	where $Q_{1} = R_{1}K_{2}+2R_{1}+1$, $Q_{2} = \frac{2}{\left(\lambda_{\boldsymbol{J}}\right)_{\text{max}}}R_{1}K_{2}$, $Q_{3} = R^{2}_{2}$, $G_{1} = V_{2}\left(t^{\text{on}}_{k}\right) - V_{2\infty}$ and $G_{2} = V_{2\infty}$. For the final conclusion derivation, further defining $M_{1} = Q_{1}s+Q_{2}G_{1}$ and $M_{2} = Q_{1}m+Q_{2}G_{2}+Q_{3}$. Correspondingly, it derives the final differential equation of $\|\boldsymbol{e}_{u}\|^{2}$, expressed as follows:
	\begin{equation}\label{finaldiff}
		\begin{aligned}
			\boldsymbol{e}^{\text{T}}_{u}\dot{\boldsymbol{e}}_{u} &\le
			M_{1}e^{-\beta\left(t-t^{\text{on}}_{k}\right)}+M_{2}
		\end{aligned}
	\end{equation}
	Integrating equation (\ref{finaldiff}) from both side, we have:
	\begin{equation}\label{ee}
		\|\boldsymbol{e}_{u}\|^{2} \le \frac{M_{1}}{\beta}\left[1-e^{-\beta\left(t-t^{\text{on}}_{k}\right)}\right]+M_{2}\left(t-t^{\text{on}}_{k}\right)
	\end{equation}
	This result indicates that the $\|\boldsymbol{e}_{u}\|^{2}$ trajectory will not beyond the function expressed in equation (\ref{ee}).
	Such a trigger condition in equation (\ref{ee}) is hard to evaluate. We further make a relaxation to it by considering the equality condition instead. Considering a time instant $t^{*}_{2}$, where $t^{*}_{2}$ satisfies:
	\begin{equation}\label{eee}\begin{aligned}
			&\frac{M_{1}e^{\beta t^{\text{on}}_{k}}}{\beta}\left(e^{-\beta t^{\text{on}}_{k}}-e^{\beta t^{*}_{2}}\right)+M_{2}\left(t^{*}_{2}-t^{\text{on}}_{k}\right) \\=& se^{-\beta t^{*}_{2}} + m
		\end{aligned}
	\end{equation}
	such a condition will be satisfied earlier than the original one, i.e., $t^{*}_{2} \le t^{\text{off}}_{k}$. Rearranging the condition in equation (\ref{eee}), it is equivalent to:
	\begin{equation}\label{eqeq}\begin{aligned}
			&\left[\left(\frac{M_{1}e^{\beta t^{\text{on}}_{k}}}{\beta}+s\right)e^{-\beta t^{\text{on}}_{k}}-M_{2}t^{\text{on}}_{k}\right]\\
			&\quad - \left[\left(\frac{M_{1}e^{\beta t^{\text{on}}_{k}}}{\beta}+s\right)e^{-\beta t^{*}_{2}}-M_{2}t^{*}_{2}\right] = se^{-\beta t^{\text{on}}_{k}} + m
		\end{aligned}
	\end{equation}
	For further analysis, constructing a function $\mathfrak{H}\left(t\right)$ as :
	\begin{equation}
		\mathfrak{H}\left(t\right) = \left(\frac{M_{1}e^{\beta t^{\text{on}}_{k}}}{\beta}+s\right)e^{-\beta t}-M_{2}t
	\end{equation}
	Taking the time-derivative of $\mathfrak{H}\left(t\right)$, we have:
	\begin{equation}
		\dot{\mathfrak{H}}\left(t\right) = -\left(M_{1}e^{\beta t^{\text{on}}} + s\beta\right)e^{-\beta t} - M_{2}
	\end{equation}
	thus, the time-derivative $\dot{\mathfrak{H}}\left(t\right) < 0$ is hold for $t\in\left(0,+\infty\right)$. It can be observed that the decreasing rate of $\mathfrak{H}\left(t\right)$ will becoming slower. The equality in (\ref{eqeq}) can be written as:
	\begin{equation}\label{condition}
		\mathfrak{H}\left(t^{\text{on}}_{k}\right) - \mathfrak{H}\left(t^{*}_{2}\right) = se^{-\beta t^{\text{on}}_{k}} + m
	\end{equation}
	For the left hand side of (\ref{condition}), a relaxation can be made through first-order linearization. Considering the following trigger condition:
	\begin{equation}\begin{aligned}
			T_{2} &= \inf_{T_{2}\ge t^{\text{on}}_{k}}\left\{ |\dot{\mathfrak{H}}\left(t^{\text{on}}_{k}\right)|\left(T_{2}-t^{\text{on}}_{k}\right) = se^{-\beta t^{\text{on}}_{k}} + m\right\}\\
		\end{aligned}
	\end{equation}
	For each $t^{\text{on}}_{k}$, $se^{-\beta t^{\text{on}}_{k}}+m$ can be regarded as a constant, hence
	we have $T_{2} \le t^{*}_{2}$.
	Define $\Delta\tau = t^{\text{off}}_{k} - t^{\text{on}}_{k}$, it derives:
	\begin{equation}\label{tau}\begin{aligned}
			\Delta\tau &\ge t^{*}_{2}-t^{\text{on}}_{k} \ge T_{2}-t^{\text{on}}_{k}  \\
			& = \frac{se^{-\beta t^{\text{on}}_{k}}+m}{\left(M_{1}e^{\beta t^{\text{on}}_{k}}+s\beta\right)e^{-\beta t^{\text{on}}_{k}}+M_{2}}\\
			& = \frac{se^{-\beta t^{\text{on}}_{k}}+m}{s\beta e^{-\beta t^{\text{on}}_{k}}+M_{1}+M_{2}}
		\end{aligned}
	\end{equation}
	In conclusion, it can be observed that a specific lower boundary of the trigger inter-event time is obtained through the relaxation of the exponential function, which is strongly related to the design parameter. This completes the proof of theorem [\ref{T3}]. 
\end{proof}

\subsection{MIET of Turn-on Trigger Mechanism}\label{mieton}

\begin{theorem}\label{T4}
	The Minimum Inter-event time (MIET) between $t^{\text{off}}_{k}$ and $^{\text{on}}_{k+1}$ will be bounded, and the MIET will not only exist but strongly related to the system parameters.
\end{theorem}
\begin{proof}

	To further discuss the Minimum Inter-Event Time of the turn-on trigger mechanism (\ref{trigon}), we start from the whole process during the time interval $t^{\text{on}}_{k}$ to $t^{\text{on}}_{k+1}$.
	Define the performance evaluation function of $V_{2}\left(t\right)$ as $S_{2}\left(t\right)$, owing to the characteristic of the trigger condition (\ref{trigon}), $V_{2}\left(t\right)$ trajectory will strictly remain beneath the performance evaluation function $S_{2}\left(t\right)$. 
	
	Intuitively, the MIET of the trigger condition (\ref{trigon}) will be strongly related to $S_{2}\left(t^{\text{off}}_{k}\right) - V_{2}\left(t^{\text{off}}_{k}\right)$, as it determines the "distance" for the trigger condition function need to "go". Accordingly, to further derive the boundary of the MIET for turn-on trigger condition, i.e., $\left(t^{\text{on}}_{k+1} - t^{\text{off}}_{k}\right)_{\text{min}}$, we consider the minima of $S_{2}\left(t^{\text{off}}_{k}\right) - V_{2}\left(t^{\text{off}}_{k}\right)$ firstly.
	
	Defining a natural performance envelope for $V_{2}\left(t\right)$ trajectory, expressed as follows:
	\begin{equation}
		S_{2}\left(t\right) = \left(S_{2}\left(0\right) - S_{2\infty}\right)e^{-\beta t} + S_{2\infty}
	\end{equation}
	where $S_{2}\left(0\right)$ represents the initial value of $S_{2}\left(t\right)$, $S_{2\infty}$ stands for the terminal asymptote value of $S_{2}\left(t\right)$. It should be noticed that $S_{2}\left(0\right) > V_{2}\left(0\right)$ and $S_{2\infty} > V_{2\infty}$ should be satisfied. 
	Correspondingly, for arbitrary trigger time instant $t^{\text{on}}_{k}$, $S_{2}\left(t^{\text{on}}_{k}\right) \ge V_{2}\left(t^{\text{on}}_{k}\right)$ will be hold. Accordingly, combining with the result in equation (\ref{V2CONC}), one can be obtained that for $t\in\left[t^{\text{on}}_{k},t^{\text{off}}_{k}\right]$, we have:
	\begin{equation}
		\begin{aligned}
			V_{2}\left(t\right)\le\left(S_{2}\left(t^{\text{on}}_{k}\right)-V_{2\infty}\right)e^{-\beta\left(t-t^{\text{on}}_{k}\right)}+V_{2\infty}
		\end{aligned}
	\end{equation} 
	where $S_{2}\left(t^{\text{on}}_{k}\right) = \left(S_{2}\left(0\right)-S_{2\infty}\right)e^{-\beta t^{\text{on}}_{k}}+S_{2\infty}$ can be derived. Thus, we can yield that:
	\begin{equation}
		\begin{aligned}
			&S_{2}\left(t\right)-V_{2}\left(t\right)\\
			\ge& S_{2}\left(t\right)-\left[\left(S_{2}\left(t^{\text{on}}_{k}\right)-V_{2\infty}\right)e^{-\beta\left(t-t^{\text{on}}_{k}\right)}+V_{2\infty}\right]\\
			= &\left[S_{2}\left(0\right)-S_{2\infty}\right]e^{-\beta t}+S_{2\infty}\\
			\quad&-\left[\left(S_{2}\left(t^{\text{on}}_{k}\right)-V_{2\infty}\right)e^{-\beta\left(t-t^{\text{on}}_{k}\right)}+V_{2\infty}\right]\\
			=& \left[S_{2}\left(0\right)-S_{2\infty}\right]e^{-\beta t^{\text{on}}_{k}}e^{-\beta\left(t-t^{\text{on}}_{k}\right)}+S_{2\infty}\\
			\quad&-\left[\left(S_{2}\left(t^{\text{on}}_{k}\right)-V_{2\infty}\right)e^{-\beta\left(t-t^{\text{on}}_{k}\right)}+V_{2\infty}\right]\\
		\end{aligned}
	\end{equation} 
	Further, it should be noticed that $\left[S_{2}\left(0\right)-S_{2\infty}\right]e^{-\beta t}+S_{2\infty}$ can be rewritten as $\left[S_{2}\left(0\right)-S_{2\infty}\right]e^{-\beta t^{\text{on}}_{k}}e^{-\beta\left(t-t^{\text{on}}_{k}\right)}+S_{2\infty}$. Therefore, for $t\in\left[t^{\text{on}}_{k},t^{\text{off}}_{k}\right]$, it concludes the result.
	\begin{equation}
		\begin{aligned}
			&S_{2}\left(t\right)-V_{2}\left(t\right)\\
			\ge&\left[-S_{2\infty}+V_{2\infty}\right]e^{-\beta\left(t-t^{\text{on}}_{k}\right)}+S_{2\infty}-V_{2\infty}\\
			=& \left(S_{2\infty}-V_{2\infty}\right)\left(1-e^{-\beta\left(t-t^{\text{on}}_{k}\right)}\right)
		\end{aligned}
	\end{equation}
	Note that $S_{2\infty}-V_{2\infty}>0$ is a constant. Meanwhile, $t^{\text{off}}_{k}-t^{\text{on}}_{k}\ge\Delta\tau$ has been proved in the previous section, as stated in equation (\ref{tau}). This indicates that $S_{2}\left(t^{\text{off}}_{k}\right)-V_{2}\left(t^{\text{off}}_{k}\right) \ge N_{k}$ will be satisfied, where $N_{k}$ is a defined defined as:
	\begin{equation}
		N_{k} = \left(S_{2\infty}-V_{2\infty}\right)\left(1-e^{-\beta\Delta\tau}\right)
	\end{equation}
	$N_{k}$ can be regarded as a minimum decreasing of the trajectory between $t^{\text{on}}_{k}$ and $t^{\text{off}}_{k}$.
	Subsequently, according to the conclusion in equation (\ref{Vbound2}), we have:
	\begin{equation}\label{Vtbound}
		V_{2}\left(t\right) \le V_{2}\left(t^{\text{off}}_{k}\right) + \frac{a^{2}_{0}}{4}\left(t-t^{\text{off}}_{k}\right)^{2} + \frac{a^{2}_{0}}{2}\left(t-t^{\text{off}}_{k}\right)C_{m}
	\end{equation} 
	In conclusion, the system behavior between $t^{\text{off}}_{k}$ and $t^{\text{on}}_{k+1}$ can be concluded as: $V_{2}\left(t\right)$ trajectory will increase from $V_{2}\left(t^{\text{off}}_{k}\right)$. Meanwhile, the performance envelope will decrease from $S_{2}\left(t^{\text{off}}_{k}\right)$. Then $S_{2}\left(t\right)-\delta_{m}$ trajectory will intersect with $V_{2}\left(t\right)$ at some time instant, satisfying the trigger condition (\ref{trigon}).

	Further, since $S_{2}\left(t^{\text{off}}_{k}\right) - V_{2}\left(t^{\text{off}}_{k}\right) \ge N_{k}$,
	considering a time instant $t_{3}>t^{\text{off}}_{k}$, such that the following relationship is hold.
	\begin{equation}\label{t3}
		t_{3} = 	\inf_{t>t^{\text{off}}_{k}}\left\{
		S_{2}\left(t_{3}\right)-\delta_{m} = \frac{a^{2}_{0}}{4}\left(t_{3}-t^{\text{off}}_{k}+W_{1}\right)^{2}	
		\right\}
	\end{equation}
	where $W_{1}$ is defined as $W_{1} = \frac{2}{a_{0}}\sqrt{S_{2}\left(t^{\text{off}}_{k}\right)-N_{k}}$. It can be observed that if the right hand side term is larger, the trigger condition will be satisfied earlier, thus, it can be infered that $t_{3}\le t^{\text{on}}_{k+1}$ will be hold.
	For the right hand side of the equation, it can be further rewritten as:
	\begin{equation}
		\frac{a^{2}_{0}}{4}\left[\left(t_{3}-t^{\text{off}}_{k}\right)^{2}+W^{2}_{1}\right]+\frac{a^{2}_{0}}{2}\left(t_{3}-t^{\text{off}}_{k}\right)W_{1}
	\end{equation}
	Notably, $\frac{a^{2}_{0}}{4}W^{2}_{1} = S_{2}\left(t^{\text{off}}_{k}\right)-N_{k}$ is hold. Further, the trigger condition in equation (\ref{t3}) can be further rearranged as follows:
	\begin{equation}
		\begin{aligned}
			\inf_{t^{*}_{3}>t^{\text{off}}_{k}}\left\{
			S_{2}\left(t^{*}_{3}\right)-S_{2}\left(t^{\text{off}}_{k}\right)=\frac{a^{2}_{0}}{4}\left(t^{*}_{3}-t^{\text{off}}_{k}\right)^{2}+W_{2}		
			\right\}
		\end{aligned}
	\end{equation}
	where $W_{2} = \frac{a^{2}_{0}}{2}\left(t^{*}_{3}-t^{\text{off}}_{k}\right)W_{1}+\delta_{m}-N_{k}$.
	Therefore, we have $t^{*}_{3}= t_{3}$. 
	
	Moreover, a vital conclusion should be noticed that $S_{2}\left(t\right)$ is an exponential function. Thus, its decreasing rate will becoming slower as it approaching the asymptote value $S_{2\infty}$. Considering another time instant $T_{3}$ satisfies the following trigger condition, expressed as:
	\begin{equation}
		\begin{aligned}
			\inf_{T_{3}>t^{\text{off}}_{k}}\left\{
			-|\dot{S}_{2}\left(t^{\text{off}}_{k}\right)|\left(T_{3}-t^{\text{off}}_{k}\right)=\frac{a^{2}_{0}}{4}\left(T_{3}-t^{\text{off}}_{k}\right)^{2}+W_{2}
			\right\}
		\end{aligned}
	\end{equation}
	it can be inferred that $T_{3}\le t^{*}_{3}$ will be always hold. Define $\Delta\mu_{2} = T_{3}-t^{\text{off}}_{k}$, $\Delta\mu_{2}$ can be regarded as a solution to the equation expressed as follows:
	\begin{equation}\label{eqn}
		\frac{a^{2}_{0}}{4}\Delta^{2}\mu_{2} + \left[|\dot{S}_{2}\left(t^{\text{off}}_{k}\right)|+\frac{a^{2}_{0}}{2}W_{1}\right]\Delta\mu_{2} + \delta_{m}-N_{k} = 0
	\end{equation}
	Notably, $a^{2}_{0}>0$ and $|\dot{S}_{2}\left(t^{\text{off}}_{k}\right)| > 0$ is hold eventually. Let $\delta_{m}-N_{k}<0$ is  always hold, the real number positive solution to the equation (\ref{eqn}) will be always exists. According to the characteristic of the parabola curve, it can be inferred that $\Delta\mu_{2}$ will achieve its minima when $|\dot{S}_{2}\left(t^{\text{off}}_{k}\right)| + \frac{a^{2}_{0}}{2}W_{1}$ reaches its maxima.
	Define $\Delta\tau_{2} = t^{\text{on}}_{k+1} - t^{\text{off}}_{k}$, we have the final conclusion as: 
	\begin{equation}
		T_{3}\le t^{*}_{3}=t_{3}\le t^{\text{on}}_{k+1}
	\end{equation}
	Which infers that $\Delta\tau_{2} \ge \Delta\mu_{2}$.
	This indicates that the inter-event time of the turn-on trigger is lower bounded, which is related to the system parameter. As a result, this completes the proof of theorem [\ref{T4}].
\end{proof}

\bibliographystyle{IEEEtran}
\bibliography{ietcppc}

\begin{thebibliography}{10}
\providecommand{\url}[1]{#1}
\csname url@samestyle\endcsname
\providecommand{\newblock}{\relax}
\providecommand{\bibinfo}[2]{#2}
\providecommand{\BIBentrySTDinterwordspacing}{\spaceskip=0pt\relax}
\providecommand{\BIBentryALTinterwordstretchfactor}{4}
\providecommand{\BIBentryALTinterwordspacing}{\spaceskip=\fontdimen2\font plus
\BIBentryALTinterwordstretchfactor\fontdimen3\font minus
  \fontdimen4\font\relax}
\providecommand{\BIBforeignlanguage}[2]{{%
\expandafter\ifx\csname l@#1\endcsname\relax
\typeout{** WARNING: IEEEtran.bst: No hyphenation pattern has been}%
\typeout{** loaded for the language `#1'. Using the pattern for}%
\typeout{** the default language instead.}%
\else
\language=\csname l@#1\endcsname
\fi
#2}}
\providecommand{\BIBdecl}{\relax}
\BIBdecl

\bibitem{bechlioulis_adaptive_2009}
\BIBentryALTinterwordspacing
C.~P. Bechlioulis and G.~A. Rovithakis, ``Adaptive control with guaranteed
  transient and steady state tracking error bounds for strict feedback
  systems,'' \emph{Automatica}, vol.~45, no.~2, pp. 532--538, 2009. [Online].
  Available: \url{https://doi.org/10.1016/j.automatica.2008.08.012}
\BIBentrySTDinterwordspacing

\bibitem{wei2018learning}
\BIBentryALTinterwordspacing
C.~Wei, J.~Luo, H.~Dai, and G.~Duan, ``Learning-based adaptive attitude control
  of spacecraft formation with guaranteed prescribed performance,'' \emph{IEEE
  transactions on cybernetics}, vol.~49, no.~11, pp. 4004--4016, 2018.
  [Online]. Available: \url{https://doi.org/10.1109/TCYB.2018.2857400}
\BIBentrySTDinterwordspacing

\bibitem{hu2018adaptive}
\BIBentryALTinterwordspacing
Q.~Hu, Y.~Shi, and X.~Shao, ``Adaptive fault-tolerant attitude control for
  satellite reorientation under input saturation,'' \emph{Aerospace Science and
  Technology}, vol.~78, pp. 171--182, 2018. [Online]. Available:
  \url{https://doi.org/10.1016/j.ast.2018.04.015}
\BIBentrySTDinterwordspacing

\bibitem{WANG2022}
\BIBentryALTinterwordspacing
K.~Wang, T.~Meng, W.~Wang, R.~Song, and Z.~Jin, ``Finite-time extended state
  observer based prescribed performance fault tolerance control for spacecraft
  proximity operations,'' \emph{Advances In Space Research}, 2022. [Online].
  Available: \url{https://doi.org/10.1016/j.asr.2022.05.072}
\BIBentrySTDinterwordspacing

\bibitem{yong2020flexible}
\BIBentryALTinterwordspacing
K.~Yong, M.~Chen, Y.~Shi, and Q.~Wu, ``Flexible performance-based robust
  control for a class of nonlinear systems with input saturation,''
  \emph{Automatica}, vol. 122, 2020. [Online]. Available:
  \url{https://doi.org/10.1016/j.automatica.2020.109268}
\BIBentrySTDinterwordspacing

\bibitem{lei2022robust}
\BIBentryALTinterwordspacing
J.~Lei, T.~Meng, W.~Wang, C.~Yin, and Z.~Jin, ``Robust control for spacecraft
  attitude tracking under multiple physical limitations with guaranteed
  performance,'' \emph{arXiv preprint arXiv:2209.05755}, 2022. [Online].
  Available: \url{https://doi.org/10.48550/arXiv.2209.05755}
\BIBentrySTDinterwordspacing

\bibitem{li2007stabilization}
\BIBentryALTinterwordspacing
C.~Li, G.~Feng, and X.~Liao, ``Stabilization of nonlinear systems via
  periodically intermittent control,'' \emph{IEEE Transactions on Circuits and
  Systems II: Express Briefs}, vol.~54, no.~11, pp. 1019--1023, 2007. [Online].
  Available: \url{https://doi,org/10.1109/TCSII.2007.903205}
\BIBentrySTDinterwordspacing

\bibitem{wang2016stabilization}
\BIBentryALTinterwordspacing
Q.~Wang, Y.~He, G.~Tan, and M.~Wu, ``Stabilization of linear systems via
  state-dependent intermittent control,'' in \emph{2016 35th Chinese Control
  Conference (CCC)}, 2016, pp. 1556--1561. [Online]. Available:
  \url{https://doi,org/10.1109/ChiCC.2016.7553312}
\BIBentrySTDinterwordspacing

\bibitem{gawthrop2009event}
\BIBentryALTinterwordspacing
P.~J. Gawthrop and L.~Wang, ``Event-driven intermittent control,''
  \emph{International Journal of Control}, vol.~82, no.~12, pp. 2235--2248,
  2009. [Online]. Available: \url{https://doi.org/10.1080/00207170902978115}
\BIBentrySTDinterwordspacing

\bibitem{gawthrop2011intermittent}
\BIBentryALTinterwordspacing
P.~Gawthrop, I.~Loram, M.~Lakie, and H.~Gollee, ``Intermittent control: a
  computational theory of human control,'' \emph{Biological cybernetics}, vol.
  104, no.~1, pp. 31--51, 2011. [Online]. Available:
  \url{https://doi.org/10.1007/s00422-010-0416-4}
\BIBentrySTDinterwordspacing

\bibitem{gawthrop2015intermittent}
P.~Gawthrop, H.~Gollee, and I.~Loram, ``Intermittent control in man and
  machine,'' \emph{Event-based control and signal processing}, pp. 281--350,
  2015.

\bibitem{loram2011human}
\BIBentryALTinterwordspacing
I.~D. Loram, H.~Gollee, M.~Lakie, and P.~J. Gawthrop, ``Human control of an
  inverted pendulum: is continuous control necessary? is intermittent control
  effective? is intermittent control physiological?'' \emph{The Journal of
  physiology}, vol. 589, no.~2, pp. 307--324, 2011. [Online]. Available:
  \url{https://doi.org/10.1113/jphysiol.2010.194712}
\BIBentrySTDinterwordspacing

\bibitem{ong2022stability}
\BIBentryALTinterwordspacing
P.~Ong, G.~Bahati, and A.~D. Ames, ``Stability and safety through
  event-triggered intermittent control with application to spacecraft orbit
  stabilization,'' \emph{arXiv preprint arXiv:2204.03110}. [Online]. Available:
  \url{https://doi.org/10.48550/arXiv.2204.03110}
\BIBentrySTDinterwordspacing

\bibitem{walls_globally_2005}
\BIBentryALTinterwordspacing
R.~J. Wallsgrove and M.~R. Akella, ``Globally stabilizing saturated attitude
  control in the presence of bounded unknown disturbances,'' \emph{J. Guid.
  Control. Dynam}, vol.~28, no.~5, pp. 957--963, 2005. [Online]. Available:
  \url{https://doi.org/10.2514/1.9980}
\BIBentrySTDinterwordspacing

\bibitem{wu2018event}
\BIBentryALTinterwordspacing
B.~Wu, Q.~Shen, and X.~Cao, ``Event-triggered attitude control of spacecraft,''
  \emph{Advances in Space Research}, vol.~61, no.~3, pp. 927--934, 2018.
  [Online]. Available: \url{https://doi.org/10.1016/j.asr.2017.11.013}
\BIBentrySTDinterwordspacing

\bibitem{lei2022singularity}
J.~Lei, T.~Meng, W.~Wang, H.~Li, and Z.~Jin, ``Singularity-avoidance prescribed
  performance attitude tracking of spacecraft,'' \emph{arXiv preprint
  arXiv:2206.12761}, 2022.

\bibitem{velasco2009lyapunov}
\BIBentryALTinterwordspacing
M.~Velasco, P.~Mart{\'\i}, and E.~Bini, ``On lyapunov sampling for event-driven
  controllers,'' in \emph{Proceedings of the 48h IEEE Conference on Decision
  and Control (CDC) held jointly with 2009 28th Chinese Control Conference},
  2009, pp. 6238--6243. [Online]. Available:
  \url{https://doi.org/10.1109/CDC.2009.5400541}
\BIBentrySTDinterwordspacing

\bibitem{khalil2015nonlinear}
H.~K. Khalil, \emph{Nonlinear control}, 2015, vol. 406.

\end{thebibliography}

\end{document}